\documentclass[12pt]{article}

\usepackage[english]{babel}

\usepackage[letterpaper,top=2cm,bottom=2cm,left=3cm,right=3cm,marginparwidth=1.75cm]{geometry}
\usepackage[utf8]{inputenc} 
\usepackage[T1]{fontenc}    
\usepackage{hyperref}       
\usepackage{natbib}
\usepackage{url}            
\usepackage{booktabs}       

\usepackage{colortbl}
\usepackage[normalem]{ulem}
\usepackage{graphicx}
\graphicspath{{figures/}}
\usepackage{subcaption}
\usepackage{amsfonts}
\usepackage{amsmath}
\usepackage{mathrsfs}
\usepackage{amsthm}
\usepackage{authblk}
\usepackage{xcolor}
\usepackage{bbm}
\usepackage{dsfont}
\usepackage{setspace}
\usepackage{mathtools}
\usepackage{algorithm}
\usepackage{algpseudocode}
\usepackage{adjustbox}
\usepackage{threeparttable}
\usepackage{wrapfig}

\usepackage{amssymb}
\usepackage{xspace}
\usepackage{enumitem}   
\usepackage{tikz}
\usetikzlibrary{positioning, arrows.meta, shadows.blur}
\usepackage{tabularx}
\usepackage{makecell}
\usepackage{array}
\newlength{\wildfirepanelheight}
\newtheorem{assumption}{Assumption}

\newtheorem{theorem}{Theorem}
\newtheorem{proposition}{Proposition}

\newtheorem{remark}{Remark}
\newtheorem{lemma}{Lemma}

\newcommand*{\eg}{\emph{e.g.}{}}
\newcommand*{\ie}{\emph{i.e.}{}}

\newcommand{\indep}{\perp \!\!\! \perp}

\renewcommand{\d}{\mathrm{d}}

\title{Beyond Prediction: Conformal Inference for\\Latent Distributional Parameters}
\author{Minxing Zheng}
\author{Wenbin Zhou}
\author{Shixiang Zhu}
\affil{Carnegie Mellon University}
\date{}

\begin{document}
\maketitle

\begin{abstract}
Many prediction problems seek to infer an unobserved, instance-specific parameter that governs the distribution of an observable response, even though the latent parameter is unavailable for both historical and future instances. We develop \texttt{LatentCP}, a prior-free conformal framework that constructs uncertainty sets for latent distributional parameters using only observed context--response pairs and a specified forward model. The method first constructs a conformal prediction set in the observable response space and then retains candidate latent parameters according to the probability their induced response distributions assign to that set. This inversion provides finite-sample marginal coverage without requiring latent calibration labels, a unique inverse mapping, or knowledge of the latent mixing distribution. Because latent-set efficiency depends nonmonotonically on the response-space miscoverage rate, we further introduce a multilevel procedure that aggregates normalized incompatibility scores across several response sets and selects the aggregation distribution using an independent tuning sample. Across synthetic experiments, \texttt{LatentCP} maintains nominal latent coverage under weak forward identification, observational nonidentifiability, and latent heterogeneity and multimodality, where empirical-Bayes, likelihood-based, and proxy-label conformal methods can substantially under-cover. On a California wildfire real dataset, it produces spatially adaptive uncertainty sets for latent fire intensity. Independent tuning improves efficiency, while multilevel aggregation provides additional gains
when different response levels contain complementary information, without sacrificing validity. 

\end{abstract}

\section{Introduction}
\label{sec:introduction}

In many prediction problems, the observable outcome is not the ultimate object of interest. What we really want to understand is the hidden mechanism that makes the outcome likely. For example, a reading from an air-quality sensor informs the underlying pollutant concentration \citep{maag2018survey}; a user's rating provides evidence about user preferences \citep{koren2009matrix}; and an observed wildfire event reveals information about the underlying ignition risk in the future \citep{jain2020review}. These problems share a common structure: we observe a random response given the context, but some latent mechanism governing the distribution of that response remains hidden. We refer to this hidden quantity as a \emph{latent distributional parameter} \citep{efron2024empirical}. Quantifying the uncertainty of latent parameters is important because it characterizes the mechanism that may continue to generate future outcomes. In wildfire risk prediction, for example, a prediction of next week's fire count supports short-term preparedness, whereas uncertainty about the underlying ignition-risk regime can inform long-term strategies, such as proactive de-energization, vegetation management, resource allocation, and long-term infrastructure hardening \citep{chen2026large}.

Learning such a latent parameter is difficult for several practical reasons. First, its true value is typically unavailable not only for a new instance but also for the historical instances used to train and calibrate a model. A wildfire dataset may contain weather conditions and realized fire counts, but it does not provide ground-truth labels for the underlying risk regime at every location and time. Consequently, one cannot directly compare estimated latent parameters with their true values and calibrate their errors in the usual supervised manner. Second, an observed response provides only noisy and incomplete evidence about the mechanism that produced it. The absence of a wildfire does not necessarily imply low underlying risk: a high-risk location may simply experience no ignition during a particular period. Similarly, the same user rating can arise from different preference profiles, and the same sensor reading can be consistent with different pollutant levels. Third, locations or individuals with similar recorded characteristics may still have different latent states because of unmeasured factors. Historical data can reveal population-level patterns, but it cannot eliminate the uncertainty surrounding the latent state of a particular new instance. These information limitations make a single latent estimate potentially misleading. A more honest inferential output is therefore a set containing all latent mechanisms that remain plausible given the observable evidence.

Existing methods address parts of this problem but require information that may be unavailable or unreliable. Bayesian and empirical-Bayes approaches combine the response model with a specified or estimated population distribution for the latent parameter. These methods can be effective when that distribution is accurately known, but their uncertainty estimates may be sensitive to misspecification, especially when latent states vary substantially across instances \citep{kleijn2006misspecification}. Another natural strategy is to infer a proxy latent value from each observed response and then treat these proxies as ordinary labels \citep{stuart2010inverse}. This approach can create unwarranted certainty because a noisy response may be compatible with several latent values, and selecting one inverse solution discards the remaining possibilities before uncertainty is calibrated. Standard conformal prediction avoids strong distributional assumptions and provides finite-sample coverage guarantees, but it ordinarily requires the prediction target to be observed in the calibration data \citep{lei2021conformal}. It therefore cannot be applied directly when the target is latent. The central question is whether conformal guarantees can nevertheless be used to quantify uncertainty about a latent parameter without observing latent calibration labels or specifying how latent parameters are distributed across the population.

We answer this question through a simple change of perspective: rather than calibrating uncertainty directly in the unobserved latent space, we first calibrate uncertainty in the observable response space and then transfer it back through the forward model. Our method, termed \texttt{LatentCP}, constructs a conformal prediction set for the observable response and treats it as a calibrated description of plausible behavior. Each candidate latent parameter induces a distribution over possible responses, and we retain the candidate when that distribution assigns sufficient probability to the conformal response set. 
This construction is intuitive because the observable response is the consequence through which a latent mechanism can be evaluated. A candidate mechanism should remain plausible when the outcomes it tends to generate agree with the response values supported by the data. 

A remaining question is how to choose the miscoverage rate of the conformal response set. Different levels emphasize different ranges of plausible responses and can therefore lead to substantially different latent uncertainty sets, with no single level uniformly best. Rather than committing to one level, we develop a randomized procedure that combines the evidence provided by several response sets and learns how much weight to place on each. This broader construction preserves finite-sample validity, contains the single-level method as a special case, and can produce more informative latent sets when different levels provide complementary evidence. Like the basic \texttt{LatentCP} procedure, it requires neither observed latent calibration labels nor knowledge of how latent parameters are distributed across the population, and it retains multiple latent values when they imply observationally indistinguishable behavior. 

Our numerical results support these theoretical findings. Across synthetic experiments, \texttt{LatentCP} maintains nominal latent coverage, including under weak identification, observational nonidentifiability, and latent heterogeneity, settings in which empirical-Bayes, likelihood-based, and proxy-label conformal methods can substantially under-cover. In a real-data application to California wildfires, \texttt{LatentCP} produces spatially adaptive uncertainty sets for latent fire intensity, identifies regions with greater intensity uncertainty, and yields substantially less conservative sets than those obtained by directly propagating response uncertainty.
Independent tuning and multilevel aggregation improve efficiency while preserving validity, demonstrating the value of combining response-space evidence across multiple miscoverage levels.

\paragraph{Contributions.}
Our contributions are threefold:
\begin{enumerate}
    \item We formulate the problem of conformal inference for latent distributional parameters when only contextual variables and responses are observed, without requiring latent calibration labels or knowledge of their population distribution.
    \item We develop \texttt{LatentCP}, which constructs a conformal prediction set in the observable response space and inverts it through the forward model to obtain a finite-sample valid uncertainty set for the latent parameter.
    \item We introduce a multi-level weighting procedure that combines evidence across response-space miscoverage rates and selects data-adaptive weights to improve the efficiency of the resulting latent uncertainty sets while preserving validity.
\end{enumerate}

\paragraph{Related work.}

This study lies at the intersection of conformal prediction \citep{shafer2008tutorial,lei2018distribution,angelopoulos2021gentle} and latent-variable inference \citep{bartholomew2011latent,everett2013introduction}. Conformal prediction constructs finite-sample, distribution-free prediction sets for observable responses under exchangeability. Its scope has expanded well beyond scalar regression. For example, conformal methods have been developed for functional responses and simultaneous prediction bands \citep{lei2015conformal,diquigiovanni2022conformal}, temporally dependent responses \citep{zaffran2022adaptive,xu2024conformal}, and structured outputs arising from graphs, language models, stream networks, and operator models \citep{zargarbashi2023conformal,lekeufack2024conformal,cherian2024large,NEURIPS2025_806288e6,harris2025locally}. Particularly relevant are probabilistic and generative conformal methods, which use samples from a conditional generative model to construct flexible, potentially disconnected prediction regions in the response space \citep{pmlr-v206-wang23n,zheng2024generative}. These developments provide a rich collection of procedures that can serve as the response-space component of \texttt{LatentCP}. Their inferential target, however, remains the observable response. In contrast, \texttt{LatentCP} treats the conformal response set as an intermediate object: it evaluates the probability that each candidate forward distribution assigns to this set and uses that probability to construct an uncertainty set for the latent parameter governing the response distribution.

A related literature extends conformal prediction to settings in which the desired target is not cleanly observed in the calibration data. Under label noise, \citet{einbinder2024label} characterize conditions under which conformal sets calibrated using corrupted labels remain valid for the underlying clean labels. \citet{feldman2025conformal} address noisy or missing labels through uncertain imputation and robust reweighting, while \citet{javanmardi2023conformal} generalize conformal calibration to set-valued, partially observed labels. For intrinsically unobservable quantities, conformal meta-learners construct intervals for individual treatment effects using pseudo-outcomes and stochastic dominance conditions linking observable and oracle conformity scores \citep{alaa2023conformal}. More recently, \citet{shirakawa2025conformalized} infer unobservable variables by comparing independently trained latent-variable predictors under assumptions on their residual distributions. These approaches make an unobservable target accessible through a corruption model, a candidate-label set, a pseudo-outcome, or a discrepancy between latent predictors. \texttt{LatentCP} takes a different route: it does not impute latent labels or require a proxy error whose distribution approximates the unobserved latent error. Instead, it evaluates every candidate latent parameter through its specified forward response distribution, thereby retaining all candidates that remain compatible with conformally calibrated response evidence.

Our framework is also related to hierarchical conformal prediction, although the relevant notion of hierarchy is different. \citet{dunn2023distribution} study prediction with two-layer grouped data, where observations within a group arise from a shared distribution and are therefore not exchangeable at the individual-observation level. Their subsampling and pooling procedures restore distribution-free guarantees for future observable outcomes under this dependence structure. \citet{zhou2026hierarchical} construct probabilistic conformal prediction intervals that remain informative across multiple levels of a spatial hierarchy. By comparison, the hierarchy in \texttt{LatentCP} describes the generative relationship from context to an instance-specific latent parameter and then to an observed response. The context--response pairs remain the units of conformal calibration, and the target is the latent parameter of a new unit rather than an observable response at a new or aggregated hierarchical level. Hierarchical conformal calibration could therefore be combined with \texttt{LatentCP} when the observed responses themselves exhibit grouped or multilevel dependence, but it does not by itself solve the inverse latent-inference problem considered here.

Classical latent-variable inference estimates latent states, structural parameters, or mixing distributions using likelihood-based procedures such as the EM algorithm \citep{dempster1977maximum,balakrishnan2017statistical}, Bayesian posterior computation and data augmentation \citep{tanner1987calculation,Hoffman2014}, empirical-Bayes estimation \citep{efron1973stein,soloff2025multivariate}, variational inference \citep{blei2017variational}, and simulation-based or likelihood-free inference \citep{beaumont2002approximate,lueckmann2021benchmarking}. Closest to our setting, \citet{Yash-2024} conformalize amortized variational posteriors, using the inverse variational density as a nonconformity score to obtain marginally valid latent regions. Their construction, however, assumes a prespecified prior distribution over the latent variable, and the latent variable is observed during calibration.
All of these approaches posit a prior or mixing distribution over the latent parameter and target its posterior, either assuming correctly specified prior and mixing distribution, as in the Bayesian and variational-based constructions above, or estimating that distribution from the population, as in empirical Bayes.
\texttt{LatentCP} retains the forward family linking latent parameters to responses but neither specifies nor estimates the population mixing distribution. It instead constructs a frequentist prediction set for the random latent parameter of a new instance. Thus, the forward model supplies the structural information needed to transfer coverage, while conformal calibration accounts for the unknown population-level heterogeneity using only observed context--response pairs.

The multilevel construction in \texttt{LatentCP} is connected to inference with $e$-values. An $e$-value is a nonnegative statistic whose expectation under the relevant null hypothesis is at most one \citep{vovk2021values}. This property has motivated conformal $e$-prediction, which supports, among other applications, anytime-valid prediction, data-dependent coverage selection, and prediction under ambiguous ground truth \citep{gauthier2025values}. Related ideas appear in conformalized decision-risk assessment, where an $e$-value-based conformal radius yields inverse certificates for the reliability of candidate decisions \citep{zhou2026conformalized}. In \texttt{LatentCP}, the normalized probability that a candidate forward distribution assigns outside a conformal response set is an $e$-value-like incompatibility score: at the true latent parameter, its expectation is bounded by one. This score is not obtained by converting a conformal $p$-value into an $e$-value. Rather, its first-moment bound follows from combining response-space conformal coverage with the candidate forward law. The $e$-value perspective then explains why incompatibility evidence from multiple response-space miscoverage rates can be averaged and thresholded without sacrificing the target latent-space coverage guarantee.

Finally, this aggregation is also related to methods for merging uncertainty sets. \citet{gasparin2024merging} combine arbitrarily dependent confidence or prediction sets through voting, weighting, and randomization, producing a single set with controlled coverage without requiring access to the internal construction of the input sets. \citet{alami2026set} instead introduce a set-preserving conversion from conformal $p$-values to $e$-values, enabling efficient $e$-value merging and randomization while retaining the prediction set induced by each original conformal procedure. \texttt{LatentCP} shares the motivation that a direct intersection of valid uncertainty sets need not remain valid. Its construction, however, is not a black-box merger of the fixed-level latent sets. It aggregates the underlying forward-incompatibility scores across the conformal path before applying the latent-space inclusion threshold. This additional structure allows the method to use the magnitude of the incompatibility evidence, rather than only membership votes, and to exploit the possibility that different response-space miscoverage rates distinguish different regions of the latent parameter space.

\section{Problem Setup}
\label{sec:problem-setup}

Let $X \in \mathcal{X}$ denote the context variable, let $\theta \in \Theta$ denote the latent parameter, and let $Y \in \mathcal{Y}$ denote the response variable. We consider the following hierarchical data-generating mechanism similar to ``empirical Bayes'' in \cite{efron2024empirical}:
\begin{equation}
    X \sim P_X,
    \qquad
    \theta \mid X=x \sim \pi(\cdot \mid x),
    \qquad
    Y \mid (X=x,\theta) \sim P_\theta(\cdot \mid x).
    \label{eq:hierarchical-model}
\end{equation}
The \textit{mixing distribution} $\pi(\cdot \mid x)$ captures latent heterogeneity among instances sharing the context $x$. In contrast, the \emph{forward family}
$
    \left\{
        P_\theta(\cdot \mid x) : \theta \in \Theta
    \right\}
$
specifies, for every candidate latent parameter $\theta$, the distribution of the response under context $x$. The formulation accommodates both discrete and continuous cases.

Suppose that we observe $n$ context--response pairs,
\[
    \mathcal{D}_n
    =
    \left\{
        (X_i,Y_i)
    \right\}_{i=1}^n,
\]
but the corresponding latent parameters $\theta_1,\ldots,\theta_n$ are hidden. 
Consider a new instance with covariates $X_{n+1}$, for which both the latent parameter $\theta_{n+1}$ and the response $Y_{n+1}$ are unobserved.
Given a prescribed miscoverage rate $\alpha \in (0,1)$, our goal is to construct a context-dependent uncertainty set
$
    \mathcal{U}_\alpha(X_{n+1};\mathcal{D}_n)
    \subseteq \Theta
$
for $\theta_{n+1}$
such that
\begin{equation}
    \mathbb{P}
    \left\{
        \theta_{n+1}
        \in
        \mathcal{U}_\alpha(X_{n+1};\mathcal{D}_n)
    \right\}
    \geq 1-\alpha.
    \label{eq:latent-coverage}
\end{equation}
The probability in \eqref{eq:latent-coverage} is taken jointly over the calibration sample $\mathcal{D}_n$ and the new instance $(X_{n+1},\theta_{n+1},Y_{n+1})$.
Note that $\mathcal{U}_\alpha$ is a prediction set for the random, instance-specific parameter $\theta_{n+1}$, rather than a confidence region for a fixed population parameter. See Appendix~\ref{app:set-valued-inference} for further discussion on their distinctions.
For notation simplicity, we suppress the dependence on $\mathcal{D}_n$ and write $\mathcal{U}_\alpha(X_{n+1})$ when no ambiguity arises.

In general, constructing a valid uncertainty set for an unobserved latent parameter is theoretically impossible without additional exploitable distributional structure.
In this work, we consider the following conditions:
\begin{enumerate}
    \item[$(i)$] The mixing distribution $\pi(\cdot \mid x)$ is unknown.
    \item[$(ii)$] The forward family $\{ P_\theta(\cdot \mid x):\theta \in \Theta \}$ is known and can be evaluated from.
\end{enumerate}
Compared with standard Bayesian inference, above conditions impose strictly weaker requirements because we do not require the ``prior distribution'' to be known (Condition~($i$)). Nevertheless, they suffice to address \eqref{eq:latent-coverage} because our objective is uncertainty quantification rather than full posterior estimation, with the former requiring substantially less distributional information. We provide more in-depth comparisons in Appendix~\ref{app:bayesian}. In the next subsection, we demonstrate that these two conditions are common in many real-world settings.

\subsection{Motivating Example: Wildfire Risk Prediction}
\label{sec:ex}

Consider a wildfire risk prediction problem. 
Let $X\in\mathbb R^p$ collect observed weather, vegetation, and topographic features, let $\theta\in\mathbb R^p$ denote a latent model parameter, and let $Y\in\mathbb N_0$ be the number of wildfires occurring over a future time interval of length $\Delta$. Consider the following simple Poisson model \citep{xu2023spatio, chen2026large}:
\[
    \theta\mid X=x\sim\pi(\cdot\mid x),
    \qquad
    Y\mid(X=x,\theta)
    \sim
    \operatorname{Poisson}\left(
        \Delta\lambda_{\theta}(x)
    \right),
    \qquad
    \lambda_{\theta}(x)
    =
    \exp(\theta^\top x),
\]
where $\lambda_{\theta}(x)$ is the wildfire-occurrence intensity under context $x$ and parameter $\theta$. 
In this example, the mixing distribution $\pi(\cdot \mid x)$ is unknown, whereas the forward family $\left\{ P_\theta(\cdot \mid x) : \theta \in \Theta \right\}$ is analytically induced by the Poisson process. These properties satisfy Conditions~($i$) and~($ii$), respectively.

This model captures two layers of uncertainty: 
$(i)$ Randomness in $\theta$ induces uncertainty about the underlying occurrence intensity $\lambda_{\theta}(x)$.
$(ii)$ Even for fixed $(x,\theta)$, the realized wildfire count remains random because of the Poisson process.
A conventional prediction method constructs a prediction set for $Y_{n+1}$ using the marginal response law
$
    P(\cdot\mid x)
    =
    \int_\Theta
        P_\theta(\cdot\mid x)
        ~\pi(\mathrm{d}\theta\mid x).
$
Such a set describes which wildfire outcomes may occur, but it averages over latent risk regimes and does not reveal which underlying mechanisms remain plausible. 
This distinction matters because different regimes can produce overlapping observed outcomes while implying different probabilities of severe events \citep{zhou2026hierarchical}. 
Our objective is instead to construct a set $\mathcal{U}_\alpha(x)$ of plausible latent parameters, which induces the risk interval, $\{\lambda_{\theta}(x): \theta\in\mathcal U_\alpha(x)\}$, that can inform the deployment of mitigation measures such as de-energization, vegetation management, or infrastructure hardening.

\section{Conformal Inference for Latent Parameters}
\label{sec:lawcp}

This section develops a new framework for constructing finite-sample valid uncertainty sets for latent distributional parameters, termed \texttt{LatentCP}, as illustrated in Figure~\ref{fig:architecture}. We first use conformal prediction to construct uncertainty set in the observable response space. We then evaluate how much probability each candidate forward distribution assigns to the conformal response set and retain the candidates that are sufficiently compatible with it. In this way, response-space conformal validity can be transferred to the latent parameter space without observing latent calibration labels.

\begin{figure}[!t]
    \centering
    \includegraphics[width=0.9\linewidth]{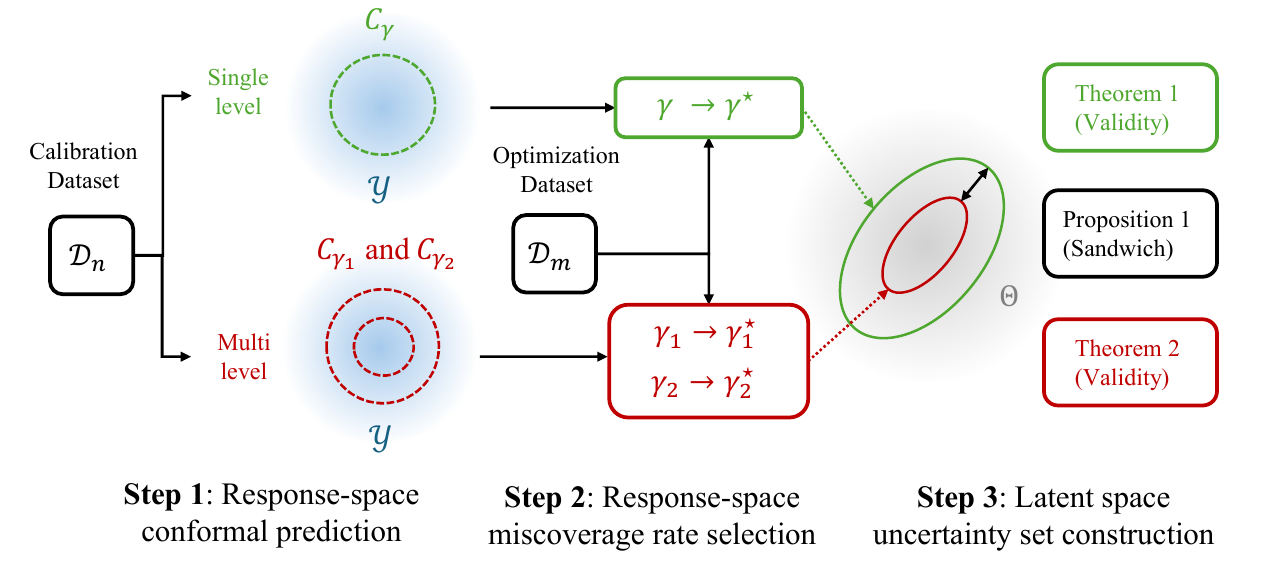}
    \caption{
    Illustration of the high-level algorithmic workflow of the proposed \texttt{LatentCP} method.
    \textbf{Step 1}: Construct the response-space conformal prediction set using the calibration dataset.
    \textbf{Step 2}: Using a separate optimization dataset, we optimize the miscoverage rate (distribution) to achieve the smallest set size in the latent space.
    \textbf{Step 3}: The latent space uncertainty set is constructed, associated with three derived theoretical guarantees.
    }
    \label{fig:architecture}
\end{figure}

\subsection{Response-Space Conformal Prediction}

Conformal prediction calibrates predictive errors on observed context--response pairs to construct a prediction set for the response at a new context, with finite-sample marginal coverage under exchangeability
\citep{vovk2005algorithmic,angelopoulos2023conformal}.
In our hierarchical setting, each observed response is generated indirectly through a sample-specific latent parameter and its associated forward distribution. The exchangeability assumption follows from a corresponding structure at the latent level: the context--latent pairs are exchangeable, and, conditional on these pairs, the responses are independently drawn from their respective forward distributions. We formalize this sampling structure below.

\begin{assumption}[Exchangeable context-latent pairs and conditional response generation]
\label{ass:conditional-latent-laws}
The pairs $\{(X_i,\theta_i)\}_{i=1}^{n+1}$ are exchangeable, with
$\theta_i\mid X_i\sim\pi(\cdot\mid X_i)$. 
Conditional on $\{(X_i,\theta_i)\}_{i=1}^{n+1}$, the responses are independent, with
\[
    Y_i\sim P_{\theta_i}(\cdot\mid X_i),
    \qquad i=1,\ldots,n+1.
\]
\end{assumption}

\begin{lemma}[Exchangeability of the observed samples]
\label{lem:observed-exchangeability}
Under Assumption~\ref{ass:conditional-latent-laws}, the observed pairs
$
    (X_1,Y_1),\ldots,(X_{n+1},Y_{n+1})
$
are still exchangeable. The proof is deferred to Appendix~\ref{app:proof-lem-obs-exch}.
\end{lemma}

Let $\widehat s:\mathcal X\times\mathcal Y\rightarrow\mathbb R$ be a fixed nonconformity score. For example, given a predictor $\widehat f$ trained independently of the calibration data, one may use
$\widehat s(x,y)=\|y-\widehat f(x)\|$. Fix a response-space miscoverage rate
$\gamma\in(0,\alpha)$, chosen independently of the calibration data and the new sample. For each calibration observation, define
$
    S_i=\widehat s(X_i,Y_i), i=1,\ldots,n,
$
and let $S_{(1)}\leq\cdots\leq S_{(n)}$ denote the corresponding order statistics. Setting
$
    k_\gamma=\left\lceil(n+1)(1-\gamma)\right\rceil,
$
we define $\widehat q_{1-\gamma}=S_{(k_\gamma)}$ if $k_\gamma\leq n$, with $\widehat q_{1-\gamma}=+\infty$ otherwise.
The resulting conformal prediction set for the observable response is
\[
    \mathcal C_\gamma(x;\mathcal{D}_n)
    =
    \left\{
        y\in\mathcal Y:
        \widehat s(x,y)\leq\widehat q_{1-\gamma}
    \right\}.
\]

This prediction set then satisfies the finite-sample marginal coverage guarantee
\begin{equation}
\label{eq:response-coverage}
    \mathbb P\left\{
        Y_{n+1}\in
        \mathcal C_\gamma(X_{n+1};\mathcal{D}_n)
    \right\}
    \geq 1-\gamma.
\end{equation}

\subsection{Conformal Forward Compatibility}

\begin{figure}[!t]
    \centering
    
    \includegraphics[width=1.\linewidth]{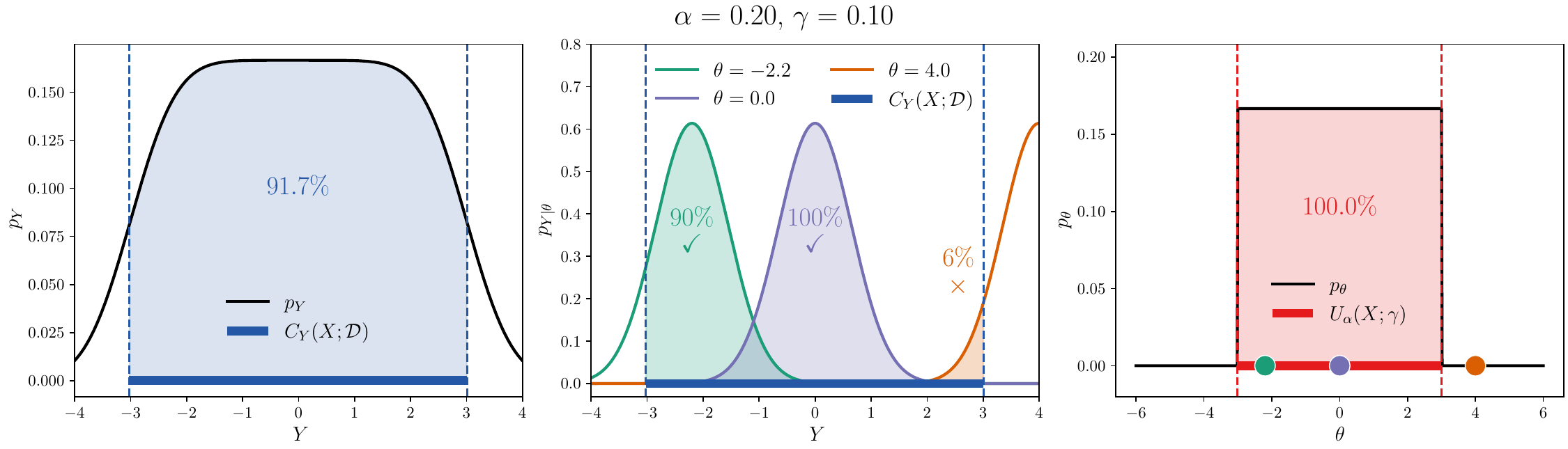}
    \caption{A stylized example of $\mathcal{U}_\alpha(x;\gamma)$ constructed according to \eqref{eq:lawcp-set} with $\alpha = 0.2$ and $\gamma = 0.10$.
    The synthetic data is generated as $\theta \sim \text{Unif}(-3, 3)$ and $Y\mid \theta \sim \mathcal{N}(\theta, 1)$ and let $x$ be an arbitrary constant.
    {\bf Left}: The marginal density $p_Y$ and the uncertainty set $\mathcal{C}_\gamma(x; \mathcal{D})$ constructed from split conformal prediction;
    {\bf Middle}: The likelihood densities $p_{Y|\theta}$ under three different configurations of $\theta$, which are marked ``$\checkmark$'' if their coverage probability with $\mathcal{C}_\gamma(x; \mathcal{D})$ is no lower than $1-\frac{\gamma}{\alpha}$, else ``$\times$''.
    {\bf Right}: The prior density $p_\theta$ and $\mathcal{U}_\alpha(x;\gamma)$, which is empirically valid (\ie, coverage probability no lower than $1-\alpha$). 
    }
    \label{fig:toy-example}
\end{figure}

Our core idea is to translate conformal validity in the response space into evidence about the latent parameter through the forward family. For a candidate $\theta\in\Theta$ and context $x$, define its \emph{forward $\gamma$-compatibility} as
\begin{equation}
\label{eq:forward-compatibility}
    p_\gamma(\theta, x)
    \coloneqq
    \mathbb P_{Y\sim P_\theta(\cdot\mid x)}
    \left\{
        Y\in\mathcal C_\gamma(x;\mathcal{D}_n)
    \right\}.
\end{equation}
Thus, $p_\gamma(\theta, x)$ is the probability mass that the candidate forward distribution $P_\theta(\cdot\mid x)$ assigns to the
$(1-\gamma)$-level conformal response set. A larger value indicates greater ``compatibility'' between the response distribution induced by $\theta$ and the predictive evidence captured by $\mathcal C_\gamma(x;\mathcal{D}_n)$.

We then define the latent uncertainty set by retaining candidates with sufficiently high forward $\gamma$-compatibility:
\begin{equation}
\label{eq:lawcp-set}
    \mathcal{U}_\alpha(x;\gamma)
    \coloneqq
    \left\{
        \theta\in\Theta:
        p_\gamma(\theta, x)
        \geq
        1-\frac{\gamma}{\alpha}
    \right\}.
\end{equation}
Equivalently, $\mathcal{U}_\alpha(x;\gamma)$ retains every latent parameter whose forward distribution assigns at most $\gamma/\alpha$ probability to responses outside the conformal set.
The threshold $\gamma/\alpha$ is chosen to transfer the response-space error bound $\gamma$ into the target latent miscoverage rate $\alpha$. 
Figure~\ref{fig:toy-example} presents a stylized example that illustrates the proposed procedure.
The following theorem formalizes this transfer.

\begin{theorem}[Finite-sample validity of \texttt{LatentCP}]
\label{thm:lawcp-validity}
Under Assumption~\ref{ass:conditional-latent-laws},
\[
    \mathbb P\left\{
        \theta_{n+1}\in
        \mathcal{U}_\alpha(X_{n+1};\gamma)
    \right\}
    \geq 1-\alpha.
\]
The proof is deferred to Appendix~\ref{app:proof-main-theorem}.
\end{theorem}

Theorem~\ref{thm:lawcp-validity} converts a conformal guarantee for the observable response into finite-sample coverage for the unobserved latent parameter, without requiring latent calibration labels or knowledge of the mixing distribution $\pi(\cdot\mid x)$. Although this guarantee holds for every fixed $\gamma\in(0,\alpha)$, the choice of $\gamma$ can substantially affect the efficiency (\ie, size) of the resulting latent uncertainty set. Specifically, decreasing $\gamma$ enlarges $\mathcal C_\gamma$, which tends to increase the forward compatibility $p_\gamma(\theta, x)$, but also raises the inclusion threshold $1-\gamma/\alpha$, requiring a candidate parameter to achieve greater compatibility to be retained. Increasing $\gamma$ reverses both effects. Because these effects work in opposite directions, the latent uncertainty set need not expand or contract monotonically with $\gamma$, and neither the smallest nor the largest admissible level is necessarily optimal.
Their relationship is further illustrated by the first row of Figure~\ref{fig:tradeoff}. Across the four synthetic settings, the average set size curve demonstrates a near-convex shape as a function of $\gamma$. An optimal $\gamma^\star$ clearly exists, and the set size improvement is nontrivial.

\begin{figure}[!t]
    \centering

    \includegraphics[width=0.80\linewidth]
        {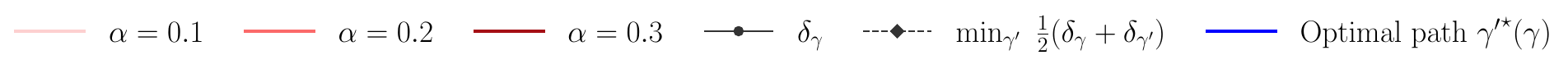}

    \vspace{0.5em}

    \begin{subfigure}[t]{0.24\linewidth}
        \centering
        \includegraphics[width=\linewidth]
            {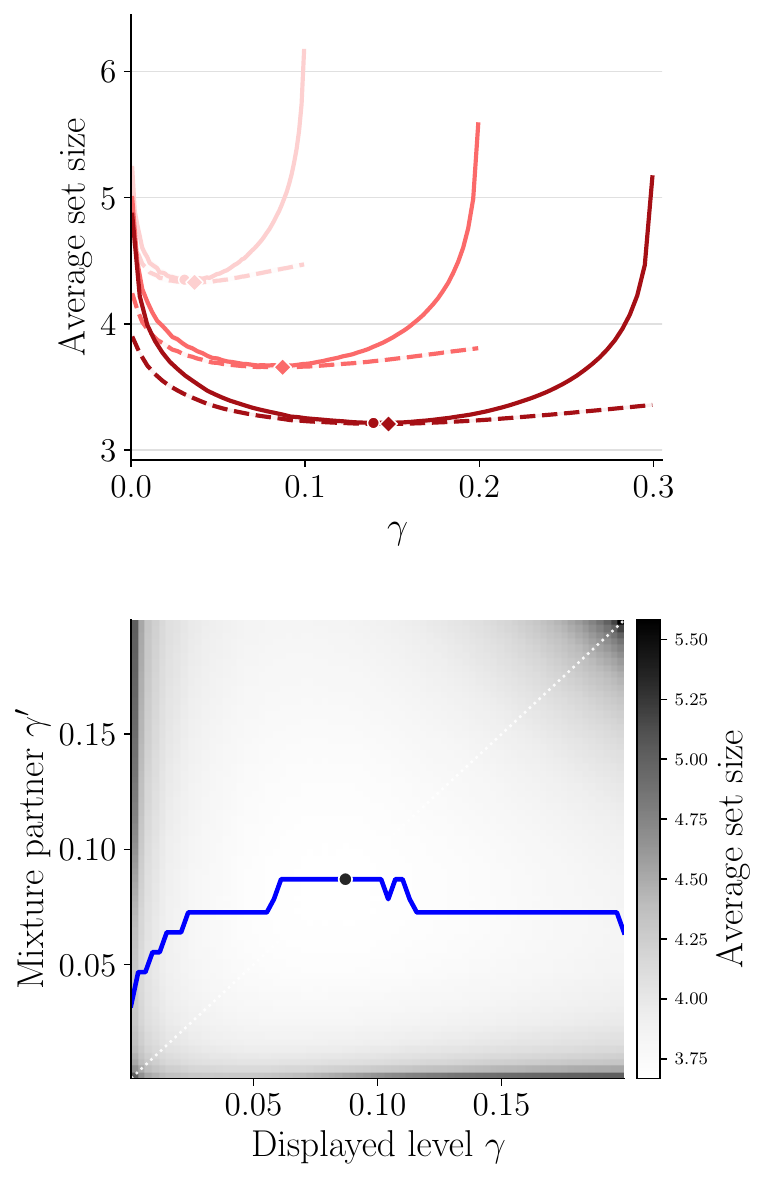}
        \caption{
            Gaussian. 
        }
        \label{fig:tradeoff-gaussian}
    \end{subfigure}
    \hfill
    \begin{subfigure}[t]{0.24\linewidth}
        \centering
        \includegraphics[width=\linewidth]
            {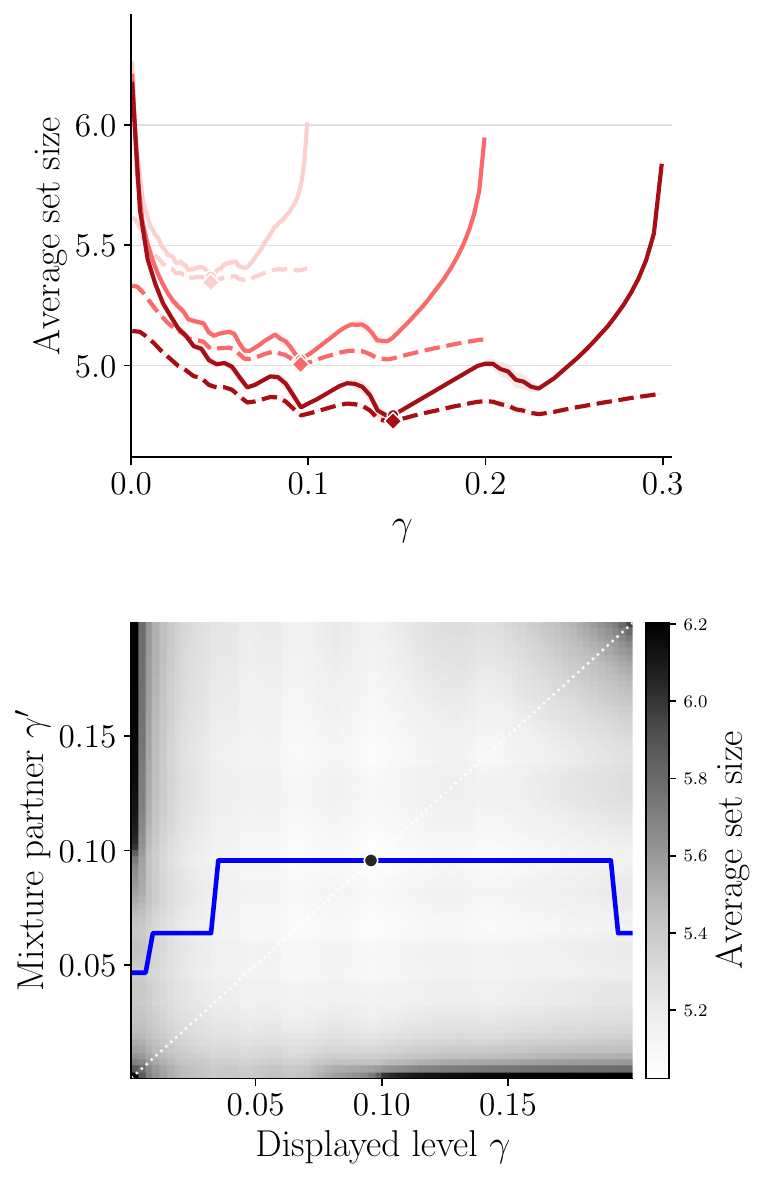}
        \caption{
            Poisson.
        }
        \label{fig:tradeoff-poisson}
    \end{subfigure}
    \hfill
    \begin{subfigure}[t]{0.24\linewidth}
        \centering
        \includegraphics[width=\linewidth]
            {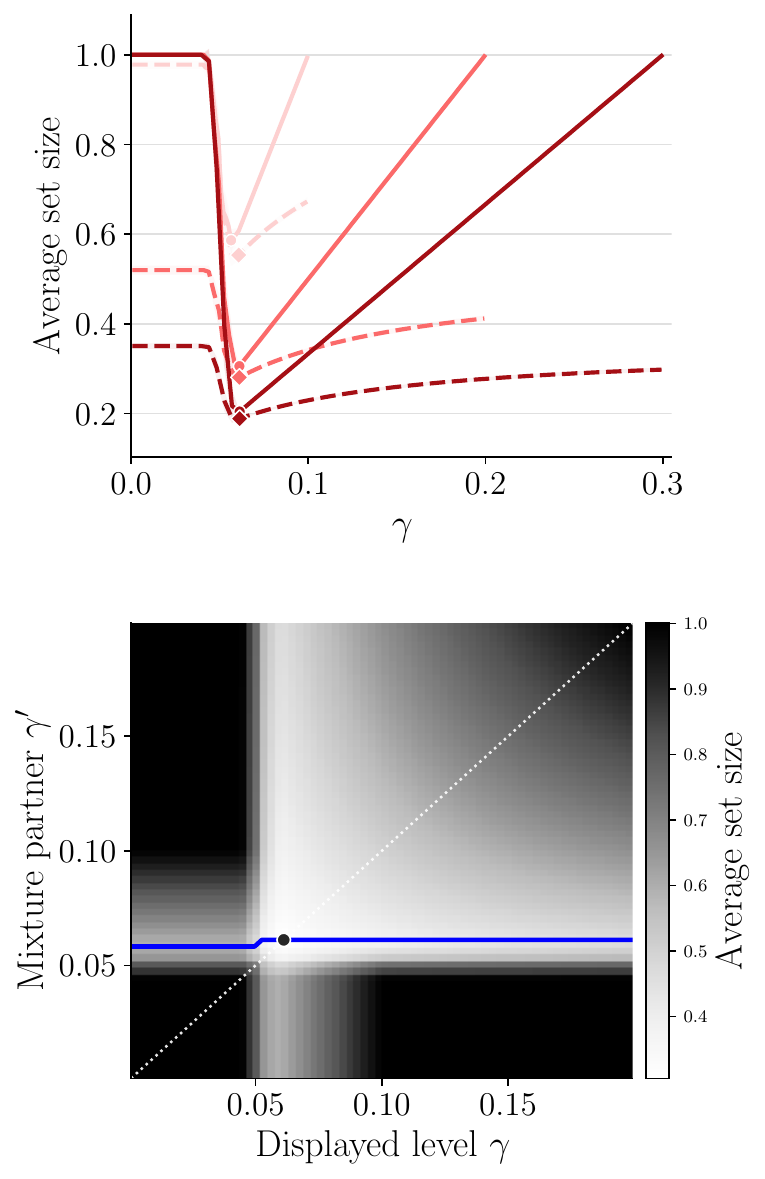}
        \caption{
            Bernoulli.
        }
        \label{fig:tradeoff-bernoulli}
    \end{subfigure}
    \hfill
    \begin{subfigure}[t]{0.24\linewidth}
        \centering
        \includegraphics[width=\linewidth]
            {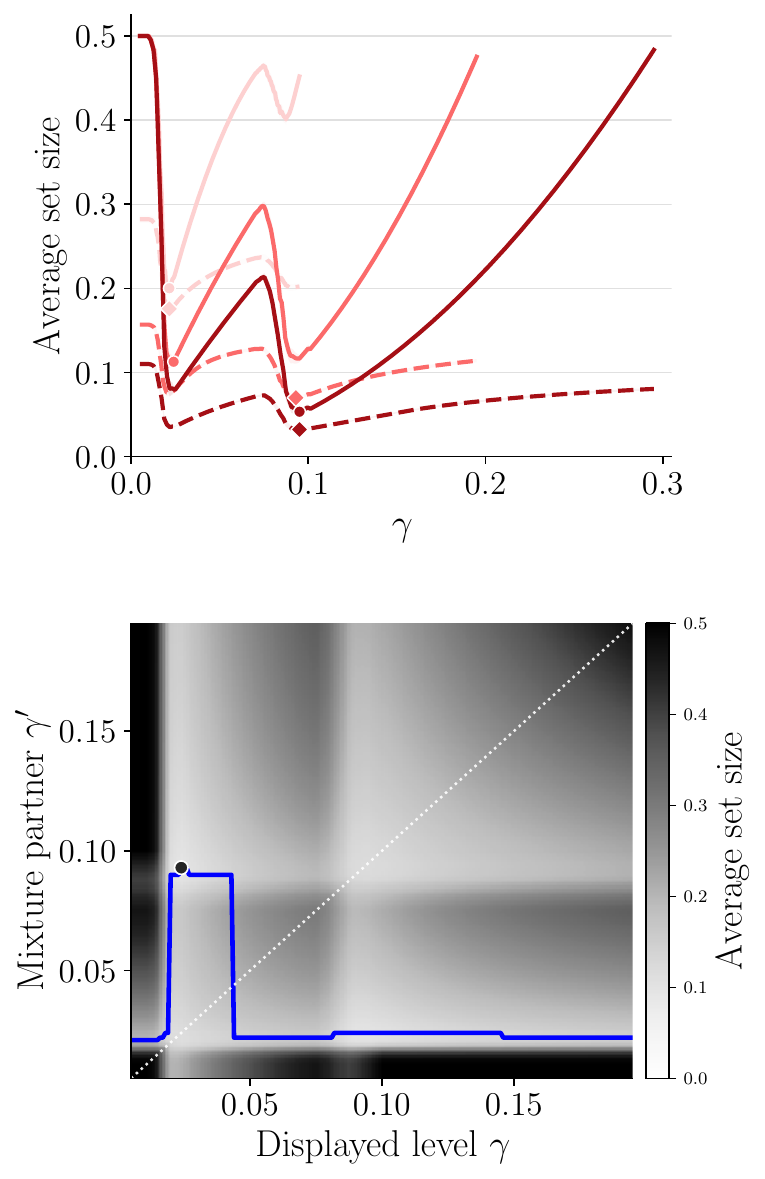}
        \caption{
            Categorical.
        }
        \label{fig:tradeoff-simplex}
    \end{subfigure}
    \caption{Effect of the auxiliary miscoverage rate $\gamma'$ on the size of the \texttt{LatentCP} uncertainty set under four synthetic data-generating mechanisms.
    \textbf{Top}: comparison of the original point-mass construction $\delta_\gamma$ with the two-point $e$-value mixture $\nu_{\gamma,\gamma'}=\frac{1}{2}\delta_\gamma+\frac{1}{2}\delta_{\gamma'}$. The mixture-set size is minimized over $\gamma'$ for each $\gamma$. Curves report the average uncertainty-set size over $60$ Monte Carlo repetitions with $3,000$ calibration samples for $\alpha\in\{0.1,0.2,0.3\}$, and solid dots mark the minimizing $\gamma^\star$.
    \textbf{Bottom}: Visualization of the corresponding two-dimensional mixture-set-size surfaces over $(\gamma,\gamma')$ for $\alpha=0.2$, with the curves indicating the minimizing $\gamma'^\star(\gamma)$.
    Each column corresponds to a distinct data-generating distribution, whose configuration details are deferred to Appendix~\ref{app:tradeoff-config}. 
    }
    \label{fig:tradeoff}
\end{figure}

\section{Multilevel-Miscoverage Uncertainty Set}
\label{sec:multilevel}

We next examine how to select the response-space miscoverage rate $\gamma$---and whether multiple miscoverage rates can be leveraged---to improve the efficiency (\ie, size) of the latent uncertainty set while preserving its coverage guarantee. Because different response miscoverage rates may extract complementary information about $\theta$, a natural approach is to combine the uncertainty sets obtained along the conformal path. A direct intersection, however, is not generally valid, since the marginal validity of the individual sets does not imply validity of their intersection. 

We therefore develop a randomized construction that aggregates normalized forward-incompatibility scores across response-space miscoverage rates before applying the latent-space inclusion threshold. The resulting procedure preserves finite-sample validity, contains every fixed-level method as a special case, and can exploit complementary evidence across levels. Its advantage is not uniform, however: aggregation may shrink the set when different levels distinguish different regions of the latent space, but may enlarge it by diluting a strong level-specific signal. We accordingly formulate the choice of miscoverage-level sampling as an optimization problem that minimizes the expected measure of the resulting latent uncertainty set.

\subsection{Randomized Miscoverage Rate Selection}
\label{subsec:multilevel-aggregation}

Recall that $p_\gamma(\theta, x)$ denotes the probability mass that the forward distribution indexed by $\theta$ assigns to the conformal response set $\mathcal{C}_\gamma(x;\mathcal D_n)$. We define the corresponding forward-incompatibility probability as
$
1-p_\gamma(\theta, x).
$
When evaluated at the true latent parameter of the future instance, this quantity satisfies
\[
\mathbb E\left[1-p_\gamma(\theta_{n+1}, X_{n+1})\right] = 
\mathbb P\left\{Y_{n+1}\notin \mathcal{C}_\gamma(X_{n+1};\mathcal D_n)\right\}
\leq \gamma.
\]
where the equality follows from the law of iterated expectation over the calibration sample and the future context--latent--response triple. The final inequality follows from the marginal conformal coverage guarantee.
This observation motivates the normalized $\gamma$-incompatibility score
\[
e_\gamma(\theta,x)\coloneqq\frac{1-p_\gamma(\theta, x)}{\gamma}.
\]
At the true latent parameter,
$
\mathbb E\left[e_\gamma(\theta_{n+1},X_{n+1})\right]\leq 1.
$
Thus, $e_\gamma$ satisfies the first-moment property of an $e$-value-type statistic for assessing the compatibility of a candidate latent parameter with the response-space prediction set \citep{vovk2021values}.

Let $\Gamma\subset(0,\alpha)$ be a prespecified collection of admissible response-space miscoverage rates, and let $\mathcal P(\Gamma)$ denote the set of probability distributions supported on $\Gamma$. For any $\nu\in\mathcal P(\Gamma)$, define the \emph{randomized incompatibility score}
\[
e_\nu(\theta,x)\coloneqq\int_\Gamma e_\gamma(\theta,x)~\mathrm d\nu(\gamma)=\int_\Gamma\frac{1-p_\gamma(\theta, x)}{\gamma}~\mathrm d\nu(\gamma).
\]

We then construct the $e$-value-based latent uncertainty set as follows:
\begin{equation}
\label{eq:multi-lawcp-set}
\mathcal{U}_\alpha(x;\nu)\coloneqq\left\{\theta\in\Theta:e_\nu(\theta,x)\leq\frac{1}{\alpha}\right\}.
\end{equation}
Unlike the procedure in \eqref{eq:lawcp-set} that selects a single level $\gamma$ and returns $\mathcal{U}_\alpha(x;\gamma)$, the construction in \eqref{eq:multi-lawcp-set} produces a set that combines the scores $e_\gamma(\theta,x)$ across all levels in the support of $\nu$ before applying the inclusion threshold $1/\alpha$.

\begin{theorem}[Validity of the randomized construction]
\label{thm:multilevel-validity}
Suppose that the conditions of Theorem~\ref{thm:lawcp-validity} hold. Let $\nu\in\mathcal P(\Gamma)$ be fixed independently of the calibration sample and the future instance. Then
\[
\mathbb P\left\{\theta_{n+1}\in\mathcal{U}_\alpha(X_{n+1};\nu)\right\}\geq 1-\alpha.
\]
The proof is deferred to Appendix~\ref{app:proof-multilevel-validity}.
\end{theorem}

\begin{remark}[Recovery of the fixed-level construction]
The original fixed-level uncertainty set is recovered by taking $\nu=\delta_\gamma$, where $\delta_\gamma$ denotes the point mass at $\gamma$. In this case,
$
e_{\delta_\gamma}(\theta,x)=e_\gamma(\theta,x),
$
and consequently,
\[
\mathcal{U}_\alpha(x;\delta_\gamma)=\mathcal{U}_\alpha(x;\gamma).
\]
Thus, the randomized construction contains every fixed-level procedure as a special case and introduces no additional coverage penalty.
\end{remark}

\subsection{Optimal Weighting for Set-Size Minimization}

To understand the effect of the randomization in \eqref{eq:multi-lawcp-set}, consider a finite collection
$
\Gamma=\{\gamma_1,\ldots,\gamma_K\}
$
and a discrete distribution
$
\nu=\sum_{k=1}^K w_k\delta_{\gamma_k},~ w_k>0,~ \sum_{k=1}^K w_k=1.
$
The randomized incompatibility score is then
$
e_\nu(\theta,x)=\sum_{k=1}^K w_k e_{\gamma_k}(\theta,x).
$
The corresponding aggregated set lies between the intersection and the union of the fixed-level sets. 

\begin{proposition}[Sandwich property]
\label{prop:aggregation-sandwich}
For every $x$ and every realization of $\mathcal D_n$,
\[
\bigcap_{k=1}^K\mathcal{U}_\alpha(x;\gamma_k)\subseteq\mathcal{U}_\alpha(x;\nu)\subseteq\bigcup_{k=1}^K\mathcal{U}_\alpha(x;\gamma_k).
\]
The proof is deferred to Appendix~\ref{app:proof-agg-sandwich}. 
\end{proposition}

\begin{figure}[!t]
    \centering
    \includegraphics[width=.4\linewidth]{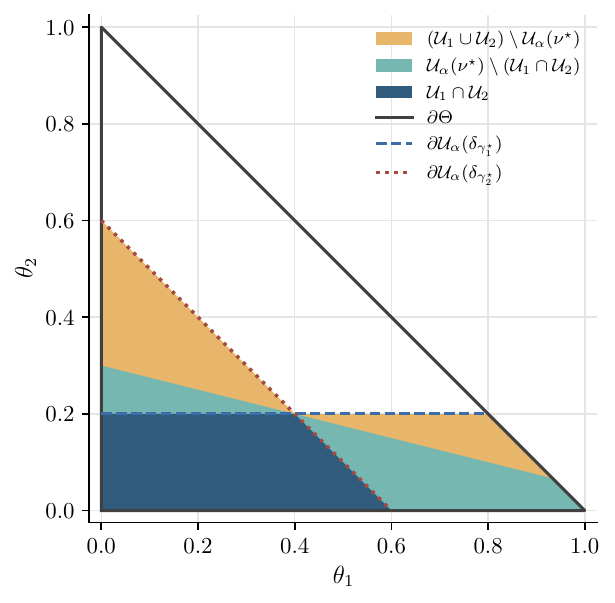}
    \caption{
    Visualization of the constructed latent uncertainty set with \texttt{LatentCP} with the categorical example used in Figure~\ref{fig:tradeoff}(d).
    Let $\mathcal{U}_i$ denote the single-level uncertainty constructed by using the point mass $\delta_{\gamma_i^\star}$ according to \eqref{eq:lawcp-set}, the $e$-value-based mixture set $\mathcal{U}_\alpha(\nu^\star)$ strictly contains the intersection set $\mathcal{U}_1 \cap \mathcal{U}_2$, but is also contained within the union set $\mathcal{U}_1 \cup \mathcal{U}_2$, which directly illustrates Proposition~\ref{prop:aggregation-sandwich}.
    }
    \label{fig:sandwich}
\end{figure}

Figure~\ref{fig:sandwich} illustrates how multilevel aggregation changes the geometry of the latent uncertainty set using a two-level categorical example with the same setup as Figure~\ref{fig:tradeoff}(d). Comparing the aggregated set with the two fixed-level sets $\mathcal{U}_1$ and $\mathcal{U}_2$ demonstrates the strict inclusions established in Proposition~\ref{prop:aggregation-sandwich}.
One interpretation of the sandwich behavior is that multilevel aggregation balances two opposing effects. ($i$) Different response levels may provide complementary evidence against different regions of the latent parameter space. If one level assigns large incompatibility scores to one group of candidates and another level does so for a different group, their weighted average may exceed $1/\alpha$ on both groups, making the aggregated set smaller than either fixed-level set. ($ii$) Aggregation may also dilute a strong level-specific signal. A candidate that is strongly incompatible at one level may receive small scores at the remaining levels, causing its averaged score to fall below $1/\alpha$ and the candidate to be retained despite its exclusion by the most informative level. As a result, $\mathcal{U}_\alpha(x;\nu)$ need not be contained in, or contain, any particular fixed-level set $\mathcal{U}_\alpha(x;\gamma_k)$.

To quantify the efficiency of an uncertainty set, let $\mu$ be a measure on $\Theta$, \eg, Lebesgue measure for a continuous parameter space, counting measure for a discrete parameter space, or a problem-specific weighted measure. Define the population expected size of the aggregated procedure as
\[
\overline{\mu}(\nu)\coloneqq\mathbb E\left[\mu\left(\mathcal{U}_\alpha(X_{n+1};\nu)\right)\right].
\]
We seek an optimal weighting distribution that solves
\begin{equation}
    \label{eq:optimization}
    \nu^\star\in\arg\min_{\nu\in\mathcal P(\Gamma)}\overline{\mu}(\nu).
\end{equation}

Because every point mass $\delta_\gamma$ belongs to $\mathcal P(\Gamma)$,
\[
\inf_{\nu\in\mathcal P(\Gamma)}\overline{\mu}(\nu)\leq\inf_{\gamma\in\Gamma}\mathbb E\left[\mu\left(\mathcal{U}_\alpha(X_{n+1};\gamma)\right)\right].
\]
Thus, at the population-oracle level, optimization over the expanded multi-level class cannot yield a larger expected set than the best fixed-level procedure. A nondegenerate distribution $\nu$ may strictly improve upon all point masses when different response levels provide sufficiently complementary information. Conversely, when one level provides the strongest discrimination throughout the relevant latent parameter space, assigning positive weight to other levels may only dilute its signal, and an optimal distribution may concentrate on a single level.

This distinction is illustrated by Figure~\ref{fig:tradeoff}. In the Gaussian, Poisson, and Bernoulli settings, the optimized multilevel and single-level constructions attain the same average set size, and the optimal pair $(\gamma^\star,\gamma'^\star)$ lies on the diagonal $\gamma'=\gamma$. Thus, the optimized mixing distribution collapses to a point mass, and the multilevel construction reduces to its single-level counterpart. By contrast, in the final setting, the optimal pair satisfies $\gamma'^\star\neq\gamma^\star$, and the resulting nondegenerate mixing distribution yields a strictly smaller average set size than the best single-level construction. This example demonstrates that multilevel aggregation can provide either no or strict efficiency improvement under different distributional conditions.

In practice, $\nu$ must be selected independently of the calibration sample used to construct the final conformal response sets. To achieve this independence, one may use an auxiliary tuning sample set to estimate $\overline{\mu}(\nu)$ and select a distribution $\widehat{\nu}$. The selected distribution is then held fixed while the conformal response sets and the final latent uncertainty set are constructed using an independent calibration sample. Conditional on the tuning sample, $\widehat{\nu}$ is fixed and independent of the calibration sample and the future instance. Therefore, Theorem~\ref{thm:multilevel-validity} applies conditionally on the tuning sample, and the resulting procedure retains its finite-sample marginal coverage guarantee.

\subsection{Computational Implementation}
\label{sec:computational-implementation}

Direct optimization over all probability distributions supported on the response-level space in \eqref{eq:optimization} is generally difficult because the class is infinite-dimensional. Although one could, in principle, use a sampling-based procedure to search over candidate distributions, such a procedure may still be computationally demanding. We restrict the weighting distribution to a simplex $\Delta_K
\coloneqq
\left\{
w\in[0,1]^K:
\sum_{k=1}^K w_k=1
\right\},$ and define
\[
\begin{aligned}
\mathcal{P}_K(\Gamma)
&\coloneqq
\left\{
\nu_{\boldsymbol{\gamma},w}
=
\sum_{k=1}^{K}w_k\delta_{\gamma_k}
:
\boldsymbol{\gamma}
=
(\gamma_1,\ldots,\gamma_K)\in\Gamma^K,\ 
w\in\Delta_K
\right\}.
\end{aligned}
\]
The optimization in \eqref{eq:optimization} therefore reduces to selecting the locations of the $K$ response-space miscoverage rates and their aggregation weights. When $K=1$, this construction recovers the original fixed-level procedure by selecting a single point mass. When $K>1$, the algorithm jointly optimizes multiple levels and aggregates their evidence through the weighted incompatibility score
\[
e_{\nu}(\theta,x)
=
\sum_{k=1}^{K} w_k e_{\gamma_k}(\theta,x).
\]
This restriction converts the infinite-dimensional optimization over distributions into a finite-dimensional optimization over support locations and simplex weights.

Once the weighting distribution has been selected and fixed, constructing the latent uncertainty set requires evaluating the forward probability in~\eqref{eq:forward-compatibility} at each of its support levels and aggregating the resulting incompatibility scores.
Each forward probability can be computed analytically, by numerical integration, or using a Monte Carlo approximation when the forward law is tractable. For a discrete latent parameter space, we evaluate this quantity at every parameter value. 
For a continuous latent parameter space, we evaluate it over a finite design $\Theta_N=\{\theta_1,\ldots,\theta_N\}$ and retain candidates according to the rule $p_\gamma(\theta,x)\geq 1-\gamma/\alpha$ for a point mass, or $e_\nu(\theta,x)\leq 1/\alpha$ for the multilevel construction. Tensor grids are convenient in low dimensions, whereas space-filling designs and local refinement near the boundary $e_\nu(\theta, x)=1/\alpha$ are more suitable in higher dimensions. The evaluations are independent across candidates and can therefore be parallelized. 

The same forward draws can be reused across all support levels and candidate support tuples. After computing their nonconformity scores once, changing a support level only changes the corresponding calibrated threshold
$
\widehat{q}_{1-\gamma_k}.
$
Moreover, the optimization over support locations need only consider the conformal breakpoints generated by the independent tuning sample set. The resulting finite search can be performed jointly for a small number of support points or through coordinate and local-search procedures when the number is larger. Appendix~\ref{app:computational-details} gives the complete procedure, its computational complexity, and a conservative enlargement that accounts for parameter discretization and numerical integration or Monte Carlo error.

\section{Experiments}
\label{sec:experiments}

This section evaluates the validity and efficiency of \texttt{LatentCP} across synthetic and real-data settings. We examine whether it maintains nominal latent coverage under weak identification, observational nonidentifiability, and model heterogeneity, and compare its set size with Bayesian, likelihood-based, noisy-label conformal, and oracle baselines. We also assess the benefits of miscoverage-rate tuning and multilevel aggregation.

\paragraph{Baseline Methods and Evaluation}

We evaluate three variants of our method: \texttt{LatentCP Base}, which uses the prespecified response-space miscoverage rate $\gamma=\alpha/2$; \texttt{LatentCP Tuned}, which selects $\gamma$ by minimizing average latent-set size on an independent tuning sample; and \texttt{LatentCP Multi-$\gamma$}, which uses the same tuning sample to jointly select two levels $\gamma_1,\gamma_2$ and a weight $w$, and aggregates them. After tuning, all conformal thresholds are recalibrated on an independent calibration sample. We compare these variants with four representative baselines. \texttt{EmpBayes} posits a parametric family for the latent
mixing distribution, estimates its parameters from the training data, and returns the smallest region with probability at least $1-\alpha$ \citep{efron1973stein,ignatiadis2019covariate}. \texttt{NPMLE} replaces the parametric mixing model with a discrete mixing distribution estimated by maximizing the marginal response likelihood over the common latent-parameter grid and returns the corresponding plug-in $1-\alpha$ latent region \citep{lindsay1995mixture,soloff2025multivariate}. \texttt{Noisy CP} computes a point-valued inverse estimate of the latent parameter from each observed context--response pair, treats these estimates as noisy calibration labels, and applies standard split conformal prediction \citep{lei2018distribution,einbinder2024label,feldman2025conformal}. Finally, \texttt{Oracle CP} applies split conformal prediction directly to the true calibration latent parameters; although unavailable in practice, it provides a reference for the efficiency attainable with fully observed latent labels \citep{lei2018distribution}. All feasible methods use the same training, tuning, calibration, and test splits and receive only the context of each test instance; neither its response nor its latent parameter is observed. We use target coverage $1-\alpha=0.9$ throughout and evaluate each method by empirical latent coverage and average uncertainty-set size, measured by cardinality for discrete parameter spaces and grid-weighted volume for continuous spaces. Figures report set size normalized by the mean Oracle CP size within each setting, and all tables and figures report averages and standard errors across independent simulation runs or real-data splits. Efficiency comparisons are made only among methods attaining the nominal coverage level. 
More details on the experimental setup are deferred to Appendix~\ref{app:experiment}.

\subsection{Synthetic Results}

We generate each synthetic dataset according to the hierarchical mechanism in Section~\ref{sec:problem-setup}. For each unit, we first sample a context $X_i$, then draw an instance-specific latent parameter $\theta_i$ from a model-specific conditional mixing distribution, and finally generate the observable response from the forward law $Y_i\mid(X_i,\theta_i)\sim P_{\theta_i}(\cdot\mid X_i)$. The latent parameters are retained only for evaluating coverage and implementing Oracle CP; they are hidden from \texttt{LatentCP} and all other feasible methods. Each dataset is divided into independent training, tuning, calibration, and test samples. The tuning sample is used to select the response-space miscoverage rate or multilevel aggregation distribution, after which the conformal response sets are recalibrated on the calibration sample. All methods are evaluated on the same test instances at target coverage $1-\alpha=0.9$. We report empirical latent coverage and average uncertainty-set size, using cardinality for discrete parameter spaces and volume for continuous spaces. Set sizes in Figures~\ref{fig:res-coverage-eff} are normalized by the mean Oracle CP size within each setting.

\begin{figure}[!t]
    \centering
    \captionsetup[subfigure]{font=footnotesize,justification=centering}
    \includegraphics[width=0.92\linewidth]{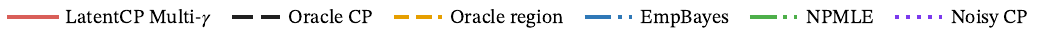}
    \begin{subfigure}[t]{0.238\linewidth}
        \centering
        \includegraphics[width=\linewidth]{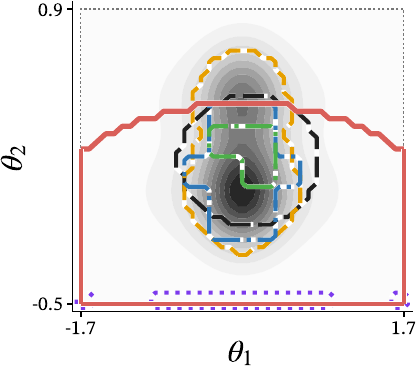}
        \caption{Gaussian.}
        \label{fig:general-2d-latent-plane-gaussian}
    \end{subfigure}\hfill
    \begin{subfigure}[t]{0.238\linewidth}
        \centering
        \includegraphics[width=\linewidth]{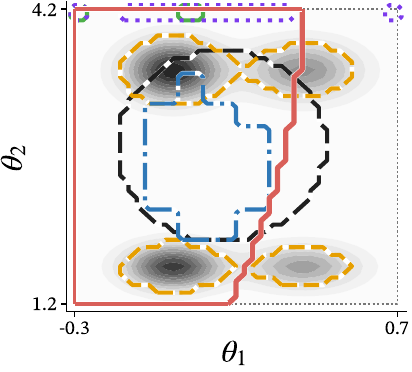}
        \caption{Negative-binomial.}
        \label{fig:general-2d-latent-plane-negbin}
    \end{subfigure}\hfill
    \begin{subfigure}[t]{0.238\linewidth}
        \centering
        \includegraphics[width=\linewidth]{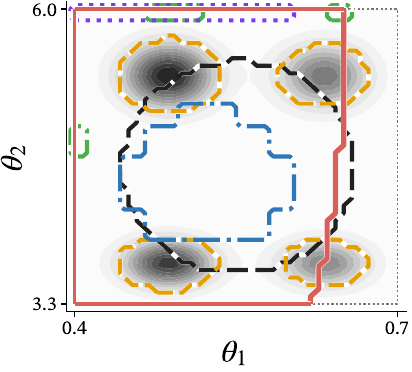}
        \caption{Beta-binomial.}
        \label{fig:general-2d-latent-plane-betabin}
    \end{subfigure}\hfill
    \begin{subfigure}[t]{0.238\linewidth}
        \centering
        \includegraphics[width=\linewidth]{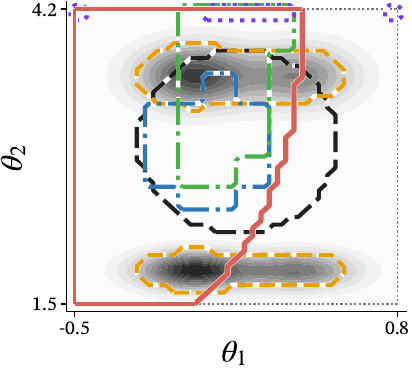}
        \caption{Student $t$.}
        \label{fig:general-2d-latent-plane-student-t}
    \end{subfigure}
    \caption{Representative two-dimensional latent uncertainty sets for four forward-model families at a fixed test context. Gray shading represents the true conditional density of the latent parameter, while colored contours delineate the uncertainty sets produced by \texttt{LatentCP} and the competing methods at the target miscoverage level $\alpha=0.3$. The panels illustrate how the methods differ in their treatment of latent dependence, weak identification, and multiple observationally compatible parameter regions.}
    \label{fig:res-latent-plane}
\end{figure}
The experiments cover both one- and two-dimensional latent parameters and are designed to isolate different forms of indirect-observation difficulty. The one-dimensional settings include a Gaussian mixture model with $Y\mid(\theta,X)\sim\mathcal{N}(\theta,0.2^2)$, a sign-nonidentifiable model with $Y\mid(\theta,X)\sim\mathcal{N}(\theta^2,0.25^2)$, a three-regime Gaussian-mixture model with crossing response distributions, and an aliased point-mass mixture with uniform contamination. These settings respectively represent latent heterogeneity, exact observational nonidentifiability, weak separation between latent regimes, and discrete aliasing. We also consider four two-dimensional forward models: Gaussian location--scale, negative-binomial, beta-binomial, and Student $t$, each with a context-dependent multimodal latent distribution. Each synthtic experiment is averaged over $50$ independent runs

\paragraph{Geometry of the latent sets.} 
Figure~\ref{fig:res-latent-plane} visualizes representative two-dimensional uncertainty sets for Gaussian, negative-binomial, beta-binomial, and Student-$t$ forward models at a fixed test context. The gray shading shows the true conditional latent density for evaluation only and is unavailable to \texttt{LatentCP}. The multilevel \texttt{LatentCP} sets reflect the geometry of forward compatibility. In contrast, empirical-Bayes, likelihood-based, and noisy-label procedures may concentrate on a single estimated mode or inverse branch, excluding other observationally compatible regions. Oracle CP applies split conformal prediction using the true latent calibration labels, whereas the Oracle region additionally assumes knowledge of the full conditional latent distribution and returns its smallest $1-\alpha$ probability region.

\begin{figure}[!t]
    \centering
    \captionsetup[subfigure]{font=footnotesize,justification=centering}

    \includegraphics[width=0.98\linewidth]{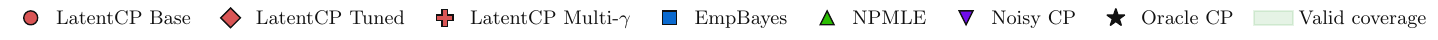}

    \begin{subfigure}[t]{0.238\linewidth}
        \centering
        \includegraphics[width=\linewidth]{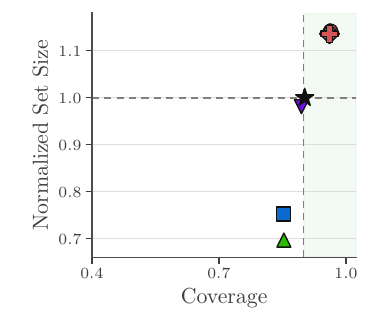}
        \caption{Gaussian mixture.
        }
        \label{fig:single-observation-mixture}
    \end{subfigure}\hfill
    \begin{subfigure}[t]{0.238\linewidth}
        \centering
        \includegraphics[width=\linewidth]{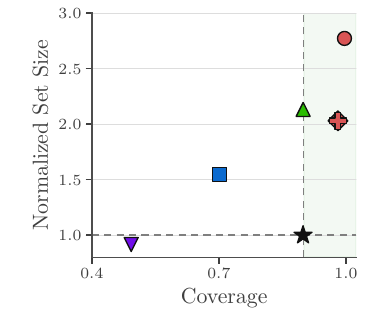}
        \caption{Sign non-identify.
        }
        \label{fig:single-observation-sign}
    \end{subfigure}\hfill
    \begin{subfigure}[t]{0.238\linewidth}
        \centering
        \includegraphics[width=\linewidth]{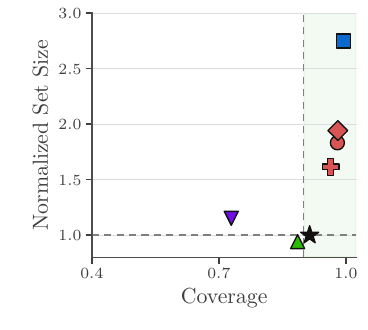}
        \caption{Regime crossing.
        }
        \label{fig:single-observation-crossing}
    \end{subfigure}\hfill
    \begin{subfigure}[t]{0.238\linewidth}
        \centering
        \includegraphics[width=\linewidth]{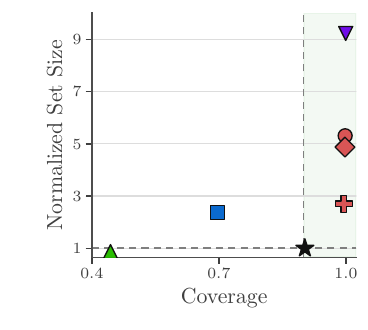}
        \caption{Aliased spikes.
        }
        \label{fig:single-observation-dilution}
    \end{subfigure}

    \vfill

    \begin{subfigure}[t]{0.238\linewidth}
        \centering
        \includegraphics[width=\linewidth]{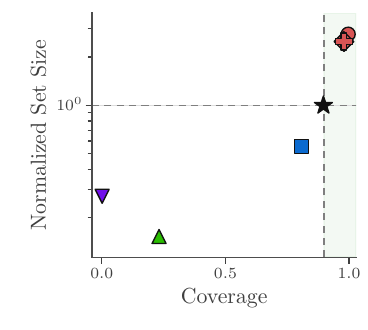}
        \caption{Gaussian.
        }
        \label{fig:selected-two-dimensional-student}
    \end{subfigure}\hfill
    \begin{subfigure}[t]{0.238\linewidth}
        \centering
        \includegraphics[width=\linewidth]{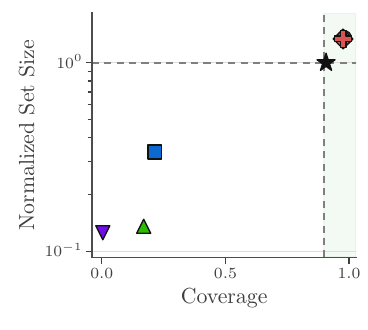}
        \caption{Negative-binomial.
        }
        \label{fig:selected-two-dimensional-betabin}
    \end{subfigure}\hfill
    \begin{subfigure}[t]{0.238\linewidth}
        \centering
        \includegraphics[width=\linewidth]{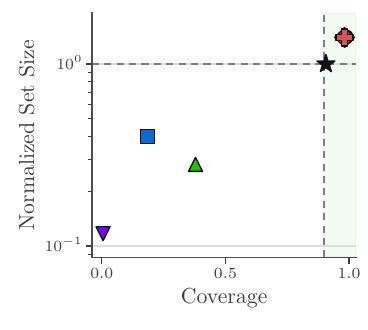}
        \caption{Beta-binomial.
        }
        \label{fig:selected-two-dimensional-location-increasing}
    \end{subfigure}\hfill
    \begin{subfigure}[t]{0.238\linewidth}
        \centering
        \includegraphics[width=\linewidth]{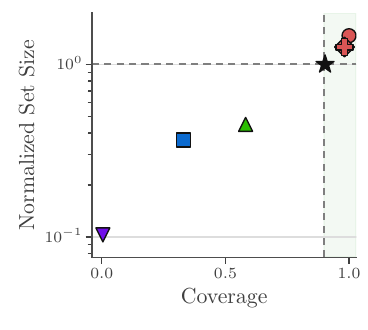}
        \caption{Student $t$.
        }
        \label{fig:selected-two-dimensional-location-decreasing}
    \end{subfigure}

    \caption{Coverage--efficiency comparison across one- and two-dimensional synthetic settings. The top row shows four one-dimensional stress tests, and the bottom row shows four two-dimensional latent-parameter models; exact forward distributions and parameter grids are provided in Appendix~\ref{app:experiment}. Each point reports mean latent coverage and set size normalized by the mean Oracle CP set size within the corresponding setting. The shaded region indicates coverage at or above the target level $1-\alpha=0.9$, while the dashed lines mark nominal coverage and normalized Oracle CP size. The vertical axes are logarithmic in the two-dimensional settings.}
    \label{fig:res-coverage-eff}
\end{figure}

\paragraph{Coverage--efficiency analysis.}
Figures~\ref{fig:res-coverage-eff} show that all \texttt{LatentCP} variants attain the target latent coverage across both one- and two-dimensional settings. The competing empirical-Bayes, likelihood-based, and noisy-label conformal methods can be efficient when their prior or inverse assumptions are favorable, but frequently under-cover under sign nonidentifiability, regime crossing, discrete aliasing, and interactions among multiple latent components. In contrast, \texttt{LatentCP} retains all parameter values whose forward distributions remain compatible with the observable response evidence. Its sets are consequently larger than those of Oracle CP, particularly when distinct latent parameters induce similar responses; this gap reflects the information lost when latent calibration labels are unavailable. Independent tuning and multilevel aggregation generally reduce set size relative to the fixed-level construction while preserving validity, with the largest gains occurring when different response-space miscoverage rates provide complementary information.

\subsection{Real Data Results}
We use Monitoring Trends in Burn Severity (MTBS) records of mapped large California wildfires 
\citep{eidenshink2007project,finco2012monitoring}, and focus on $1,689$ wildfire events recorded from $1984$--$2023$. We consider two spatial resolutions: California's 58 counties for visual illustration and $50$-km grid cells for the formal analysis. Let $s$ index a spatial unit under the chosen resolution, and divide the $40$-year study period into ten nonoverlapping four-year windows indexed by $t$. For each spatial-unit--window pair $(s,t)$, $Y_{s,t}$ is the recorded wildfire count, while $X_{s,t}\in\mathbb R^{10}$ includes standardized spatial coordinates and time, log exposure, preceding-window count and historical rate through $t-1$. We then model the observed count as
\[
Y_{s,t}\mid X_{s,t},\lambda_{s,t}
\sim \operatorname{Poisson}(\Delta_{s,t}\lambda_{s,t}),
\]
where $\Delta_{s,t}$ is area--time exposure and $\lambda_{s,t}$ is the latent
fire-occurrence rate. We fit the response model on $1984$--$2003$, tune $\gamma$ on $2004$--$2007$, calibrate on $2008$--$2015$, and evaluate on $2016$--$2023$, using $\alpha=0.10$.

We first present an illustrative county-level analysis in Figure~\ref{fig:setting}, where $s$ indexes California's 58 counties. Panel (a) shows the observed counts for 2020--2023, panels (b)--(c) show the LatentCP lower and upper intensity endpoints, and panel (d) shows interval width. The results reveal substantial spatial heterogeneity: lower endpoints remain near zero in many counties, whereas the upper endpoints and widths highlight regions with greater latent uncertainty.  Figure~\ref{fig:wildfire-overlay} reports the formal analysis using 179 retained $50$-km grid cells, providing more spatial units than the county-level aggregation. Relative to the fixed choice $\gamma=\alpha/2$, tuning reduces mean latent-set size by $11.1\%$ across the test observations. The figure distinguishes the two inferential targets: $\mathcal C_\gamma$ represents uncertainty in the future fire event count, whereas $\mathcal U_\alpha$ directly represents uncertainty in the underlying intensity after transfer through the Poisson forward model. The dashed Poisson MLE band is a standard model-based reference, obtained by substituting the fitted intensity $\widehat \lambda$ into the Poisson variance. Unlike this plug-in approximation, LatentCP incorporates uncertainty calibrated from the observed population. Because counts and intensities have different units, their relative widths on the shared nonlinear axis are for illustrative purpose. Overall, LatentCP converts response-level uncertainty into spatially adaptive uncertainty sets for the unobserved fire intensity.

\begin{figure}[!t]
    \centering
    \captionsetup[subfigure]{font=footnotesize,justification=centering}
    \setlength{\wildfirepanelheight}{0.266\linewidth}
    \begin{subfigure}[t]{0.205\linewidth}
        \centering
        \includegraphics[width=\linewidth]{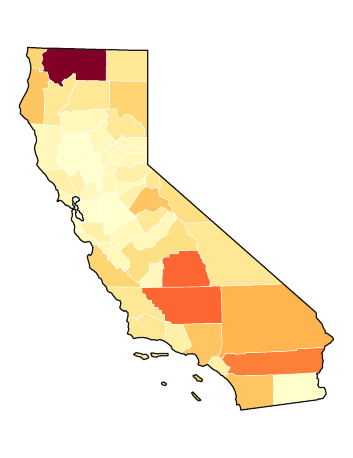}
        \includegraphics[width=\linewidth]{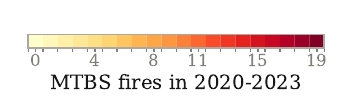}
        \caption{Observed $Y$.}
        \label{fig:setting-observed}
    \end{subfigure}\hfill
    \begin{subfigure}[t]{0.205\linewidth}
        \centering
        \includegraphics[width=\linewidth]{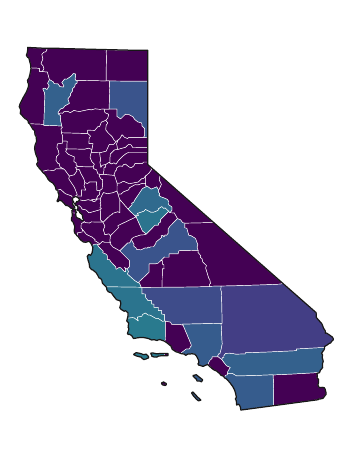}
        \phantom{\includegraphics[width=\linewidth]{figures/spatial_wildfire_real_mtbs_county_sharedscale_colorbar_observed.pdf}}
        \caption{Lower bound.}
        \label{fig:setting-lower}
    \end{subfigure}\hfill
    \begin{subfigure}[t]{0.205\linewidth}
        \centering
        \includegraphics[width=\linewidth]{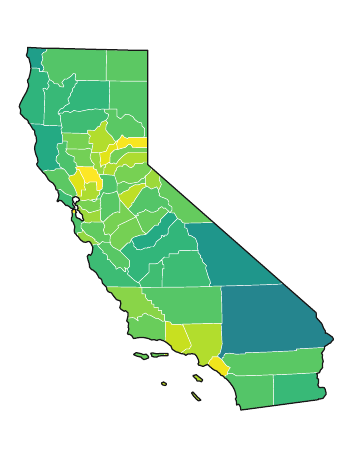}
        \phantom{\includegraphics[width=\linewidth]{figures/spatial_wildfire_real_mtbs_county_sharedscale_colorbar_observed.pdf}}
        \caption{Upper bound.}
        \label{fig:setting-upper}
    \end{subfigure}\hfill
    \begin{subfigure}[t]{0.205\linewidth}
        \centering
        \includegraphics[width=\linewidth]{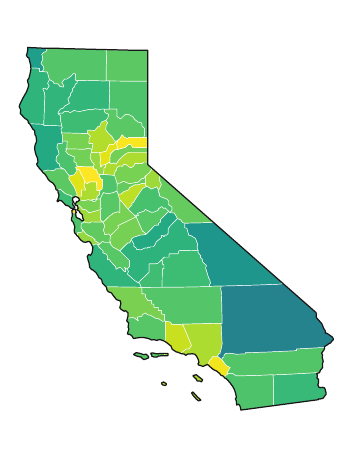}
        \phantom{\includegraphics[width=\linewidth]{figures/spatial_wildfire_real_mtbs_county_sharedscale_colorbar_observed.pdf}}
        \caption{Set size.}
        \label{fig:setting-width}
    \end{subfigure}\hfill
    \begin{subfigure}[t]{0.09\linewidth}
        \centering
        \includegraphics[height=\wildfirepanelheight]{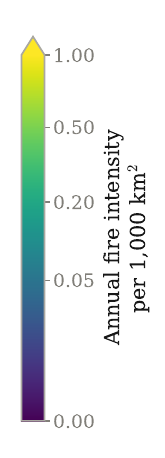}%
    \end{subfigure}%
\caption{County-level wildfire counts and latent-intensity uncertainty in California. Panel (a) shows MTBS counts for $2020$--$2023$; panels (b)--(d) show the \texttt{LatentCP} lower bound, upper bound, and interval width, respectively.}
    \label{fig:setting}
\end{figure}

\begin{figure}[!t]
    \centering
    \includegraphics[width=\linewidth]{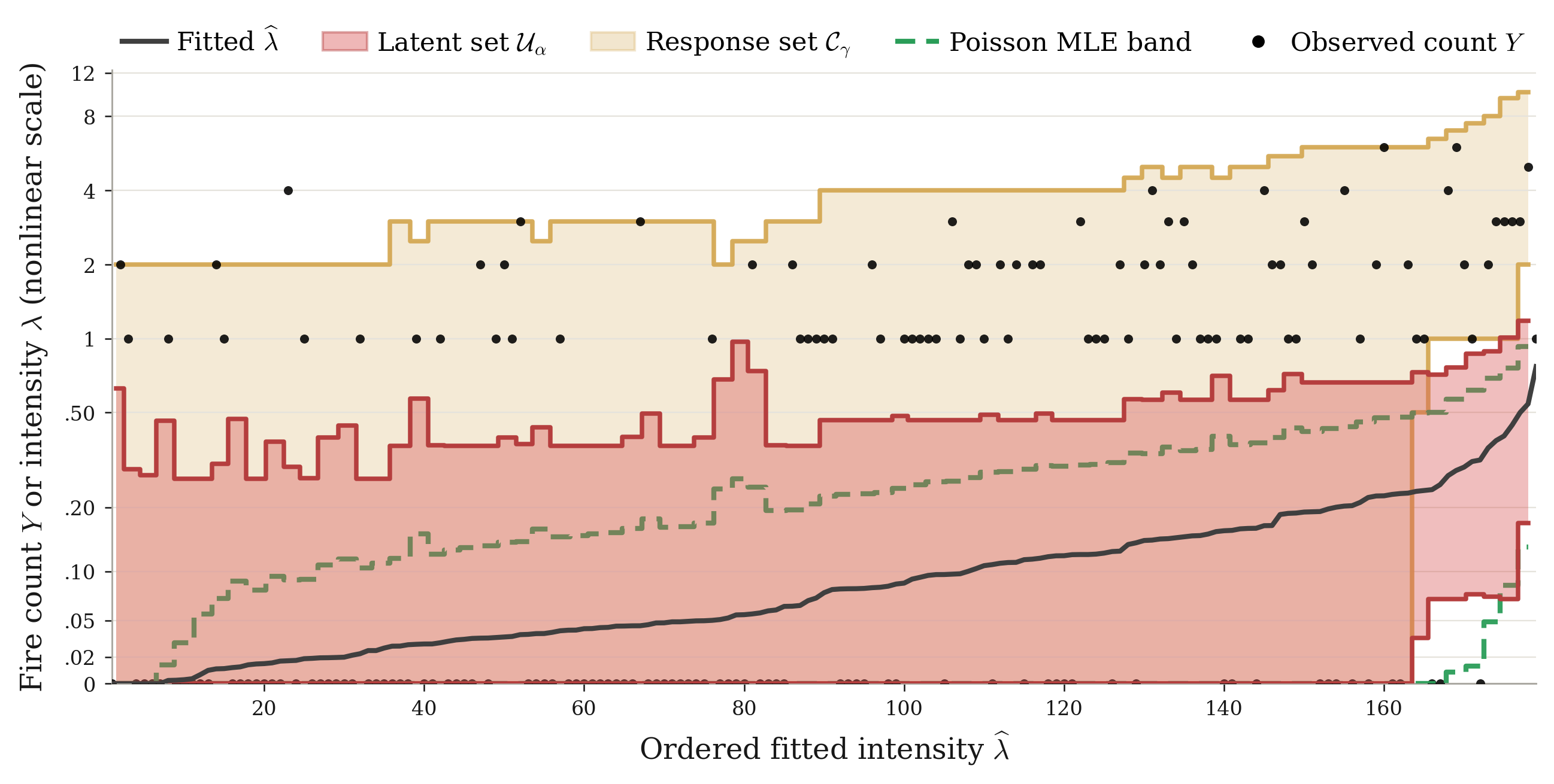}
    \caption{Response and latent uncertainty for California wildfire data.
The x-axis indexes grid cells ordered by their fitted intensities
$\widehat{\lambda}$, while the shared nonlinear y-axis displays fire counts and latent intensities. Shading area shows $\mathcal C_\gamma$ and $\mathcal U_\alpha$; points are observed counts, and dashed curves show a baseline Poisson Wald--type uncertainty band.}
\label{fig:wildfire-overlay}
\end{figure}

\section{Conclusion and Discussion}

This paper introduced \texttt{LatentCP}, a conformal framework for constructing finite-sample valid uncertainty sets for unobserved, instance-specific distributional parameters. By transferring a conformal response set through a specified forward model, \texttt{LatentCP} avoids requiring latent calibration labels, a unique inverse mapping, or knowledge of the latent mixing distribution. Its multilevel extension aggregates $e$-value-like incompatibility scores across response-space miscoverage rates to improve efficiency while preserving validity. Experiments demonstrate reliable coverage under weak identification, observational aliasing, and confounding, where methods relying on estimated latent labels or mixing distributions may under-cover.

Several limitations suggest directions for future work. The framework provides marginal rather than conditional coverage and relies on an adequately specified forward model. Extending it to forward-model misspecification, context-dependent validity, and high-dimensional or distribution-valued latent parameters constitutes important future work. Scalable compatibility evaluation, richer multilevel aggregation, and applications involving repeated or dependent responses are also promising directions. Overall, \texttt{LatentCP} provides a principled bridge from conformal prediction of observable responses to uncertainty quantification for their latent generative mechanisms.

\bibliographystyle{plainnat}
\bibliography{ref}

\newpage
\appendix

\section{Why Set-Valued Latent Inference Is Necessary}
\label{app:set-valued-inference}

This appendix formalizes two distinct obstacles to point identification of the latent distributional parameter. The first is \emph{latent heterogeneity}: the latent parameter of a new sample remains random after conditioning on its observed context. The second is \emph{observational non-identifiability}: different latent parameters may induce identical response distributions.

\subsection{Latent Heterogeneity}

Consider first the standard independent-sampling setting in which the calibration data provide information about the population but not about the particular latent draw associated with the new sample. Formally, suppose that
\[
    \theta_{n+1}\indep\mathcal{D}_n
    \mid X_{n+1}.
\]
This condition holds when the calibration samples and the new sample are independently generated according to \eqref{eq:hierarchical-model}. Conditional on $\mathcal{D}_n$ and $X_{n+1}=x$, the latent parameter therefore continues to follow $\pi(\cdot\mid x)$.

The following result bounds the probability with which any point estimator can exactly recover this latent draw.

\begin{lemma}[Impossibility of exact identification]
\label{lem:impossibility}
Suppose that
$\theta_{n+1}\indep\mathcal{D}_n\mid X_{n+1}$.
Let $\widehat\theta(\mathcal{D}_n,X_{n+1})$ be any possibly randomized estimator taking values in $\Theta$. Then
\begin{equation}
\label{eq:point-identification-bound}
    \mathbb P\left\{
        \widehat\theta(\mathcal{D}_n,X_{n+1})
        =\theta_{n+1}
    \right\}
    \leq
    \mathbb E\left[
        \sup_{t\in\Theta}
        \pi(\{t\}\mid X_{n+1})
    \right].
\end{equation}
In particular, if $\pi(\cdot\mid x)$ is non-atomic for
$P_X$-almost every $x$, then
\[
    \mathbb P\left\{
        \widehat\theta(\mathcal{D}_n,X_{n+1})
        =\theta_{n+1}
    \right\}=0
\]
for every point estimator and every calibration sample size.
\end{lemma}

\begin{proof}
Let $R$ denote any internal randomness used by the estimator, chosen independently of the data and the new sample. We may therefore write
\[
    \widehat\theta(\mathcal{D}_n,X_{n+1})
    =
    h(\mathcal{D}_n,X_{n+1},R)
\]
for some measurable function $h$.

Condition on $\mathcal{D}_n$, $X_{n+1}=x$, and $R$. Under the conditional-independence assumption, the conditional distribution of $\theta_{n+1}$ remains $\pi(\cdot\mid x)$. Hence,
\begin{align*}
    &\mathbb P\left\{
        \widehat\theta(\mathcal{D}_n,X_{n+1})
        =\theta_{n+1}
        ~\middle|~
        \mathcal{D}_n,X_{n+1}=x,R
    \right\} \\
    =&
    \pi\left(
        \left\{h(\mathcal{D}_n,x,R)\right\}
        ~\middle|~x
    \right)\\
    \leq &
    \sup_{t\in\Theta}\pi(\{t\}\mid x).
\end{align*}
Taking expectations over $\mathcal{D}_n$, $X_{n+1}$, and $R$ yields
\eqref{eq:point-identification-bound}. If
$\pi(\cdot\mid x)$ is non-atomic, then
$\pi(\{t\}\mid x)=0$ for every $t\in\Theta$, so the right-hand side of
\eqref{eq:point-identification-bound} is zero.
\end{proof}

Lemma~\ref{lem:impossibility} has a direct interpretation in the discrete case. If the largest conditional probability among the possible latent classes is $0.6$, no point estimator based only on $\mathcal{D}_n$ and $X_{n+1}$ can achieve an exact-recovery probability greater than $0.6$. A set containing multiple latent classes is therefore necessary to obtain, for example, $90\%$ coverage. In the continuous case, every singleton has probability zero, so a non-singleton set is necessary for any positive coverage rate.

The lemma concerns exact recovery before observing $Y_{n+1}$. It does not preclude accurate point estimation under a loss such as absolute or squared error, nor does it determine how large a valid uncertainty set must be. Its purpose is to explain why a singleton cannot generally satisfy the coverage objective in \eqref{eq:latent-coverage}.

\subsection{Observational Non-Identifiability}

Latent heterogeneity is distinct from non-identifiability of the forward model. For a fixed context $x$, define the observational equivalence relation
\[
    \theta\equiv_x\theta'
    \quad\Longleftrightarrow\quad
    P_\theta(\cdot\mid x)
    =
    P_{\theta'}(\cdot\mid x).
\]
Parameters in the same equivalence class generate identical observable response distributions and therefore cannot be distinguished through response data alone.

\begin{lemma}[Observational equivalence]
\label{lem:observational-equivalence}
Fix $x\in\mathcal X$ and suppose that
$\theta\equiv_x\theta'$. Let $Y$ be generated from either
$P_\theta(\cdot\mid x)$ or $P_{\theta'}(\cdot\mid x)$. Then every measurable statistic $T(Y,x)$ has the same distribution under $\theta$ and $\theta'$. In particular, for any possibly randomized test
$\varphi:\mathcal Y\times\mathcal X\rightarrow[0,1]$,
\[
    \mathbb E_{Y\sim P_\theta(\cdot\mid x)}\left[\varphi(Y,x)\mid x\right]
    =
    \mathbb E_{Y\sim P_{\theta'}(\cdot\mid x)}\left[\varphi(Y,x)\mid x\right].
\]
Consequently, no response-based test can distinguish $\theta$ from $\theta'$ with power exceeding its size.
\end{lemma}

\begin{proof}
Because $\theta\equiv_x\theta'$, the conditional probability measures
$P_\theta(\cdot\mid x)$ and $P_{\theta'}(\cdot\mid x)$ are identical. Therefore, for every measurable function $g$,
\[
\int g(y)\,P_\theta(\mathrm dy\mid x)
=
\int g(y)\,P_{\theta'}(\mathrm dy\mid x).
\]
Taking $g(y)=\varphi(y,x)$ gives the equality of rejection probabilities. More generally, applying any measurable transformation $T$ to two identically distributed responses produces identical pushforward distributions.
\end{proof}

This non-identifiability persists under repeated observations whenever the corresponding joint response laws remain identical. In particular, if responses are conditionally independent and
$P_\theta(\cdot\mid x)=P_{\theta'}(\cdot\mid x)$, then their $m$-fold product distributions are equal for every finite $m$. Repeated observations may distinguish parameters whose forward distributions are different but close; they cannot distinguish parameters whose forward distributions are exactly the same.

Our proposed procedure respects these equivalence classes. For any measurable response set $\mathcal C\subseteq\mathcal Y$, observational equivalence implies
\[
    P_\theta(Y\in \mathcal C\mid x)
    =
    P_{\theta'}(Y\in \mathcal C\mid x).
\]
Consequently, \texttt{LatentCP} applies the same inclusion decision to $\theta$ and $\theta'$. Rather than selecting an arbitrary inverse branch, it retains or excludes observationally equivalent parameters together. The resulting uncertainty set therefore represents both sample-specific latent heterogeneity and ambiguity inherent in the forward model.

\section{Connection to Other Relevant Inference Problems}
\label{app:bayesian}

In this section, we discuss the distinctions between our \texttt{LatentCP} setup and Bayesian inference (Section~\ref{app:real-bayesian}) and likelihood-based inference (Section~\ref{app:likelihood}). They represent two alternative classes of estimators that one may use to construct uncertainty sets for latent parameters, and have already been incorporated in our experiment comparison. This section dives deeper into explaining their key differences.

\subsection{Comparison to Bayesian Inference}
\label{app:real-bayesian}

\begin{table*}[t]
    \centering
    \caption{Comparison of the proposed problem with standard conformal prediction
    and Bayesian inference.}
    \label{tab:problem-comparison}
    \renewcommand{\arraystretch}{1.5}
    \setlength{\tabcolsep}{6pt}
    \begin{tabularx}{\textwidth}
    {
        >{\raggedright\arraybackslash}p{0.16\textwidth}
        >{\centering\arraybackslash}p{0.08\textwidth}
        >{\centering\arraybackslash}p{0.16\textwidth}
        >{\raggedright\arraybackslash}p{0.19\textwidth}
        >{\raggedright\arraybackslash}X
    }
        \toprule
        \textbf{Problem}
        & \textbf{Target}
        & \textbf{Info. for new unit}
        & \textbf{Required knowledge}
        & \textbf{Output} \\
        \midrule

        Standard conformal prediction
        & $Y_{n+1}$
        & $X_{n+1}$
        & Neither $P_\theta$ nor $\pi$
        & Prediction set for the response \\

        \midrule

        Bayesian posterior inference
        & $\theta_{n+1}$
        & $X_{n+1},Y_{n+1}$
        & Forward family and $\pi$
        & Posterior credible set \\

        \midrule

        \makecell[l]{Bayesian or \\ empirical-Bayes}
        & $\theta_{n+1}$
        & $X_{n+1}$
        & Specified or estimated $\pi(\cdot\mid x)$
        & Prior-predictive region \\

        \midrule

        \textbf{LatentCP}
        & $\theta_{n+1}$
        & $X_{n+1}$ only
        & Forward family, but not $\pi$
        & Frequentist marginal prediction set \\

        \bottomrule
    \end{tabularx}
\end{table*}

The hierarchical model in \eqref{eq:hierarchical-model} has a familiar Bayesian interpretation, but our inferential target and information structure are different. Suppose that $P_\theta(\cdot \mid x)$ admits a density or probability mass function $p_\theta(\cdot \mid x)$. For fixed $\theta$, $p_\theta(y\mid x)$ specifies the forward response distribution. After observing $(x,y)$, the same quantity, viewed as a function of $\theta$, serves as the likelihood for the latent parameter. If both the mixing law $\pi(\cdot \mid x)$ and the response $y$ from the instance of interest were available, Bayes' rule would yield
\begin{equation}
    \pi(\mathrm{d}\theta \mid x,y)
    =
    \frac{
        p_\theta(y\mid x)
        \pi(\mathrm{d}\theta \mid x)
    }{
        \displaystyle
        \int_\Theta
        p_{\vartheta}(y\mid x)
        \pi(\mathrm{d}\vartheta \mid x)
    }.
    \label{eq:bayesian-posterior}
\end{equation}
Thus, the forward family in our formulation plays the same mathematical role as the likelihood family in a hierarchical Bayesian model.

The distinction is that we construct $\mathcal{U}_\alpha(X_{n+1})$ before observing $Y_{n+1}$. Consequently, there is no unit-specific likelihood
$
    p_\theta(Y_{n+1}\mid X_{n+1})
$
with which to update uncertainty about $\theta_{n+1}$. Before observing the new response, a Bayesian predictive distribution for $\theta_{n+1}$ would be precisely the mixing law $\pi(\cdot \mid X_{n+1})$. When this law is unknown, a Bayesian or empirical-Bayes approach must specify it, estimate it, or place a hyperprior on it, typically through the marginal mixture
$
    p(y\mid x)
    =
    \int_\Theta
        p_\theta(y\mid x)
        \pi(\mathrm{d}\theta\mid x).
$
Estimating a context-dependent mixing law from indirect observations can be difficult or ill-posed, particularly when each latent draw generates only one response or when the forward family is non-identifiable \citep{ignatiadis2019covariate}.

Our objective is therefore not to approximate the posterior in \eqref{eq:bayesian-posterior}. The set $\mathcal{U}_\alpha$ does not assign posterior probabilities or relative weights to candidate latent parameters. Instead, it is a frequentist prediction set for the random latent parameter $\theta_{n+1}$, with finite-sample marginal coverage that does not require specifying or estimating $\pi(\cdot \mid x)$. This mixing-law robustness does not make the procedure fully model-free: validity continues to rely on exchangeability and on a correctly specified or conservatively approximated forward family.
See Table~\ref{tab:problem-comparison} for a comparison between Bayesian inference paradigms and our conformal prediction workflow.

\subsection{Comparison to Likelihood-Based Inference}
\label{app:likelihood}

Classical likelihood-based confidence intervals, such as Wald and profile-likelihood intervals, address a fundamentally different inferential problem \citep{pawitan2001all}. They typically posit a fixed parameter $\theta_0$ shared across repeated observations and quantify uncertainty about this common parameter using its likelihood. A Wald interval relies on the asymptotic normality of an estimator \citep{wald1943tests}, whereas a profile-likelihood interval is commonly obtained by inverting a likelihood-ratio statistic whose reference distribution is justified asymptotically \citep{murphy2000profile}. Consequently, except in special models admitting exact finite-sample inference, their nominal coverage is generally asymptotic and may be inaccurate in small samples or under weak identification, nonregularity, or model misspecification.

These methods are not directly applicable to our target. Because each historical response $Y_i$ is generated under its own latent parameter $\theta_i$, the historical observations do not constitute repeated measurements from a common parameter $\theta_{n+1}$. Furthermore, since $Y_{n+1}$ has not yet been observed, there is no unit-specific likelihood $p_\theta(Y_{n+1}\mid X_{n+1})$ from which a Wald or profile-likelihood interval for $\theta_{n+1}$ can be constructed. Although covariates may enter a conventional likelihood, they generally describe how the response distribution depends on observed predictors while the underlying parameter remains fixed. Modeling $\theta_{n+1}$ itself as a context-dependent random variable instead requires a distribution such as $\pi(\cdot\mid X_{n+1})$, returning to the hierarchical or empirical-Bayes formulations discussed above.
See Table~\ref{tab:likelihood-comparison} for a summary of their distinctions.

\begin{table*}[t]
    \centering
    \caption{Comparison of Wald intervals, profile-likelihood intervals, and LatentCP.}
    \label{tab:likelihood-comparison}
    \small
    \renewcommand{\arraystretch}{1.5}
    \setlength{\tabcolsep}{4pt}
    \begin{tabularx}{\textwidth}
    {
        >{\raggedright\arraybackslash}p{0.11\textwidth}
        >{\raggedright\arraybackslash}p{0.14\textwidth}
        >{\raggedright\arraybackslash}p{0.13\textwidth}
        >{\raggedright\arraybackslash}p{0.15\textwidth}
        >{\raggedright\arraybackslash}p{0.17\textwidth}
        >{\raggedright\arraybackslash}X
    }
        \toprule
        \textbf{Method}
        & \textbf{Target}
        & \textbf{Info. for inference}
        & \textbf{Context}
        & \textbf{Required knowledge}
        & \textbf{Output and guarantee} \\
        \midrule

        Wald interval
        & Fixed population parameter $\theta_0$
        & Observations governed by the same $\theta_0$
        & Covariates may enter the likelihood, but $\theta_0$ remains fixed
        & Parametric likelihood, standard error, and regularity conditions
        & Confidence interval with typically asymptotic coverage \\

        \midrule

        Profile-likelihood interval
        & Fixed parameter of interest $\psi_0$, possibly with nuisance parameters
        & Observations governed by the same underlying parameter
        & Covariates may enter the likelihood, but $\psi_0$ remains fixed
        & Full likelihood and regularity conditions for likelihood-ratio approximation
        & Likelihood-ratio confidence set with typically asymptotic coverage \\

        \midrule

        \textbf{LatentCP}
        & Random, unit-specific latent parameter $\theta_{n+1}$
        & $X_{n+1}$ only; $\theta_{n+1}$ and $Y_{n+1}$ are unobserved
        & The set $\mathcal{U}_\alpha(X_{n+1})$ adapts to the observed context
        & Forward family $\{P_\theta(\cdot\mid x):\theta\in\Theta\}$, but not $\pi(\cdot\mid x)$
        & Frequentist prediction set with finite-sample marginal coverage \\

        \bottomrule
    \end{tabularx}
\end{table*}

In contrast, \texttt{LatentCP} is designed specifically for the random, instance-specific target $\theta_{n+1}$. Its uncertainty set adapts to the observed context through both the conformal response set $\mathcal{C}_\gamma(X_{n+1};\mathcal{D}_n)$ and the context-dependent forward family $\{P_\theta(\cdot\mid X_{n+1}):\theta\in\Theta\}$. Under the exchangeability conditions in Assumption~\ref{ass:conditional-latent-laws}, it provides finite-sample marginal coverage for every calibration sample size without requiring asymptotic normality, a likelihood-ratio approximation, or knowledge of $\pi(\cdot\mid x)$. The guarantee remains marginal over the context distribution rather than conditional on each fixed value of $X_{n+1}$; nevertheless, the resulting set is context-adaptive in its construction. These distinctions make \texttt{LatentCP} more appropriate than standard Bayesian inference or fixed-parameter likelihood intervals for the inferential setting considered in this work.

\section{Technical Proofs}

\subsection{Proof of Lemma~\ref{lem:observed-exchangeability}}
\label{app:proof-lem-obs-exch}

\begin{proof}
Fix any permutation $\sigma$ of $\{1,\ldots,n+1\}$. Conditional independence implies that, after permuting the context-parameter pairs, the corresponding response vector has conditional distribution
\[
    \mathcal L\left(
        Y_{\sigma(1)},\ldots,Y_{\sigma(n+1)}
        ~\middle|~
        (X_{\sigma(i)},\theta_{\sigma(i)})_{i=1}^{n+1}
    \right)
    =
    \otimes_{i=1}^{n+1}
    P_{\theta_{\sigma(i)}}(\cdot\mid X_{\sigma(i)}).
\]
Thus, permuting the indices only reorders the same sample-specific response distributions; it does not change how the responses are jointly generated. Since the context-parameter pairs are exchangeable, the complete triples are therefore exchangeable. Marginalizing over the unobserved $\theta_i$'s establishes exchangeability of
$(X_i,Y_i)_{i=1}^{n+1}$.
\end{proof}

\subsection{Proof of Theorem~\ref{thm:lawcp-validity}}
\label{app:proof-main-theorem}

\begin{proof}
Let
$
    R_\gamma
    \coloneqq
    1 - p_\gamma(\theta_{n+1}, X_{n+1})
$
denote the probability mass that the true forward distribution assigns outside the conformal response set. By Assumption~\ref{ass:conditional-latent-laws} and
\eqref{eq:response-coverage},
\[
    \mathbb E[R_\gamma]
    =
    \mathbb P\left\{
        Y_{n+1}\notin
        \mathcal C_\gamma(X_{n+1};\mathcal{D}_n)
    \right\}
    \leq\gamma.
\]
By the definition of $\mathcal{U}_\alpha$,
$\theta_{n+1}\notin\mathcal{U}_\alpha(X_{n+1};\gamma)$ only if
$R_\gamma>\gamma/\alpha$. Markov's inequality therefore gives
\[
    \mathbb P\left\{
        \theta_{n+1}\notin
        \mathcal{U}_\alpha(X_{n+1};\gamma)
    \right\}
    \leq
    \frac{\mathbb E[R_\gamma]}{\gamma/\alpha}
    \leq\alpha.
\]
\end{proof}

\subsection{Proof of Theorem~\ref{thm:multilevel-validity}}
\label{app:proof-multilevel-validity}

\begin{proof}
By Tonelli's theorem and the fixed-level moment bound,
\[
\begin{aligned}
\mathbb E
\left[
e_\nu(\theta_{n+1},X_{n+1})
\right]
&=
\mathbb E
\left[
\int_\Gamma
e_\gamma(\theta_{n+1},X_{n+1})
\d\nu(\gamma)
\right] \\
&=
\int_\Gamma
\mathbb E
\left[
e_\gamma(\theta_{n+1},X_{n+1})
\right]
\d\nu(\gamma) \\
&\leq
\int_\Gamma 1\d\nu(\gamma)
=
1.
\end{aligned}
\]
Applying Markov's inequality gives
\[
\begin{aligned}
\mathbb P
\left\{
\theta_{n+1}
\notin
\mathcal{U}_\alpha(X_{n+1};\nu)
\right\}
&=
\mathbb P
\left\{
e_\nu(\theta_{n+1},X_{n+1})
>
\frac{1}{\alpha}
\right\} \\
&\leq
\alpha~
\mathbb E
\left[
e_\nu(\theta_{n+1},X_{n+1})
\right] \\
&\leq
\alpha.
\end{aligned}
\]
The stated coverage guarantee follows.
\end{proof}

\subsection{Proof of Proposition~\ref{prop:aggregation-sandwich}}
\label{app:proof-agg-sandwich}

\begin{proof}

We prove the two inclusions separately. First, let
\[
\theta
\in
\bigcap_{k=1}^K
\mathcal{U}_\alpha(x;\gamma_k).
\]
By the definition of the fixed-level uncertainty sets in \eqref{eq:lawcp-set},
\[
e_{\gamma_k}(\theta,x)
\leq
\frac{1}{\alpha}
\qquad
\text{for every }
k=1,\ldots,K.
\]
Because the weights are nonnegative and sum to one,
\[
e_\nu(\theta,x)
=
\sum_{k=1}^K
w_k e_{\gamma_k}(\theta,x)
\leq
\sum_{k=1}^K
w_k\frac{1}{\alpha}
=
\frac{1}{\alpha}
\sum_{k=1}^K w_k
=
\frac{1}{\alpha}.
\]
Therefore,
$
\theta\in \mathcal{U}_\alpha(x;\nu),
$
which proves
$
\bigcap_{k=1}^K
\mathcal{U}_\alpha(x;\gamma_k)
\subseteq
\mathcal{U}_\alpha(x;\nu).
$

For the second inclusion, we prove the contrapositive. Suppose that
\[
\theta
\notin
\bigcup_{k=1}^K
\mathcal{U}_\alpha(x;\gamma_k).
\]
Then $\theta$ is excluded by every fixed-level set, so
\[
e_{\gamma_k}(\theta,x)
>
\frac{1}{\alpha}
\qquad
\text{for every }
k=1,\ldots,K.
\]
It follows that
\[
e_\nu(\theta,x)-\frac{1}{\alpha}
=
\sum_{k=1}^K
w_k e_{\gamma_k}(\theta,x)
-
\frac{1}{\alpha}
\sum_{k=1}^K w_k
=
\sum_{k=1}^K
w_k
\left(
e_{\gamma_k}(\theta,x)
-
\frac{1}{\alpha}
\right)
>
0,
\]
where the strict inequality follows because every weight is positive
and every term inside the parentheses is strictly positive. Hence,
\[
e_\nu(\theta,x)
>
\frac{1}{\alpha},
\]
which implies
$
\theta\notin \mathcal{U}_\alpha(x;\nu).
$
We have therefore shown that
\[
\theta
\notin
\bigcup_{k=1}^K
\mathcal{U}_\alpha(x;\gamma_k)
\quad\Longrightarrow\quad
\theta
\notin
\mathcal{U}_\alpha(x;\nu).
\]
Taking the contrapositive gives
$
\mathcal{U}_\alpha(x;\nu)
\subseteq
\bigcup_{k=1}^K
\mathcal{U}_\alpha(x;\gamma_k).
$
Combining the two inclusions completes the proof.
\end{proof}

\section{Computational and Implementation Details}
\label{app:computational-details}

This appendix describes the practical implementation of \texttt{LatentCP}, which involves two main computational tasks. First, for a fixed response-space miscoverage level $\gamma\in(0,\alpha)$, we evaluate the forward compatibility $p_\gamma(\theta,x)$ over the latent parameter space, which in principle requires enumerating all candidate values of $\theta$. Second, using an independent tuning sample, we select either a single response level $\gamma$ or multiple levels $\gamma_1,\ldots,\gamma_K$ together with their aggregation weights $w_1,\ldots,w_K$.

\subsection{Distribution Configurations in Figure~\ref{fig:tradeoff}}
\label{app:tradeoff-config}

We consider four synthetic data-generating mechanisms without informative
context, equivalently taking $X$ to be constant. In each repetition, latent
parameters are independently sampled from the specified mixing distribution
and then used to generate the observable responses. The latent parameters are
used only for data generation and are not supplied to \texttt{LatentCP}.

For the first three settings, we use the absolute-residual nonconformity score
\[
    \widehat{s}(y)=|y-\widehat{m}|,
\]
where $\widehat{m}$ is a fitted constant response predictor held fixed during
calibration. Consequently,
\[
    \mathcal C_\gamma
    =
    [\widehat m-\widehat q_{1-\gamma},
      \widehat m+\widehat q_{1-\gamma}]
    \cap\mathcal Y,
\]
where $\widehat q_{1-\gamma}$ is the split-conformal empirical quantile.

\paragraph{Gaussian location model.}
The latent parameter follows a standard Gaussian distribution truncated to
$\Theta=[-4,4]$:
\[
    \theta\sim\mathcal N(0,1)\bigm|_{[-4,4]},
    \qquad
    Y\mid\theta\sim\mathcal N(\theta,0.5^2).
\]
For $\mathcal C_\gamma=[\ell_\gamma,u_\gamma]$, the forward compatibility is
available analytically as
\[
    p_\gamma(\theta)
    =
    \Phi\left(\frac{u_\gamma-\theta}{0.5}\right)
    -
    \Phi\left(\frac{\ell_\gamma-\theta}{0.5}\right),
\]
where $\Phi$ denotes the standard Gaussian distribution function. Set size is
measured by Lebesgue length over $[-4,4]$.

\paragraph{Poisson log-rate model.}
We generate
\[
    \theta\sim\mathcal N(0,1)\bigm|_{[-3.5,3.5]},
    \qquad
    Y\mid\theta\sim\operatorname{Poisson}\!\left(e^\theta\right),
\]
so that $\theta$ represents the logarithm of the Poisson rate. Because the
response is integer-valued, its compatibility is computed by summing the
Poisson probability mass over the conformal response set:
\[
    p_\gamma(\theta)
    =
    \sum_{y\in\mathcal C_\gamma\cap\mathbb N_0}
    \frac{\exp(-e^\theta)e^{\theta y}}{y!}.
\]
Set size is measured by Lebesgue length over $[-3.5,3.5]$.

\paragraph{Bernoulli model.}
The latent success probability and response are generated according to
\[
    \theta\sim\operatorname{Beta}(1,19),
    \qquad
    Y\mid\theta\sim\operatorname{Bernoulli}(\theta),
    \qquad
    \Theta=[0,1].
\]
For any conformal response set
$\mathcal C_\gamma\subseteq\{0,1\}$, the forward compatibility is
\[
    p_\gamma(\theta)
    =
    (1-\theta)\mathbbm 1\{0\in\mathcal C_\gamma\}
    +
    \theta\mathbbm 1\{1\in\mathcal C_\gamma\}.
\]
Set size is measured by Lebesgue length over $[0,1]$.

\paragraph{Categorical simplex model.}
Let
\[
    \Theta
    =
    \left\{
        (\theta_1,\theta_2)\in\mathbb R_+^2:
        \theta_1+\theta_2\leq 1
    \right\},
    \qquad
    \theta_0=1-\theta_1-\theta_2.
\]
We generate the latent probability vector and response as
\[
    (\theta_0,\theta_1,\theta_2)
    \sim\operatorname{Dirichlet}(55,4,1),
    \qquad
    Y\mid\boldsymbol{\theta}
    \sim\operatorname{Categorical}
    (\theta_0,\theta_1,\theta_2),
\]
and use the ordinal nonconformity score $\widehat s(y)=y$. The marginal response
probabilities are therefore
\[
    \mathbb P(Y=0)=\frac{55}{60},
    \qquad
    \mathbb P(Y=1)=\frac{4}{60},
    \qquad
    \mathbb P(Y=2)=\frac{1}{60}.
\]
For $\mathcal C_\gamma\subseteq\{0,1,2\}$, compatibility is evaluated exactly:
\[
    p_\gamma(\boldsymbol{\theta})
    =
    \sum_{j=0}^2
    \theta_j\mathbbm 1\{j\in\mathcal C_\gamma\}.
\]
Set size is measured by two-dimensional Lebesgue area over $\Theta$, whose
total area is $1/2$.

This final mechanism is constructed to illustrate a strict benefit from
multilevel aggregation. At $\alpha=0.2$, consider
\[
    \gamma_1^\star=\frac{1}{60},
    \qquad
    \gamma_2^\star=\frac{5}{60},
    \qquad
    \nu^\star
    =
    \frac{1}{2}\delta_{\gamma_1^\star}
    +
    \frac{1}{2}\delta_{\gamma_2^\star}.
\]
The corresponding population response sets are
$\mathcal C_{\gamma_1^\star}=\{0,1\}$ and
$\mathcal C_{\gamma_2^\star}=\{0\}$, which produce
\[
    \mathcal U_1
    =
    \left\{
        \boldsymbol{\theta}\in\Theta:
        \theta_2\leq\frac{1}{12}
    \right\},
    \qquad
    \mathcal U_2
    =
    \left\{
        \boldsymbol{\theta}\in\Theta:
        \theta_1+\theta_2\leq\frac{5}{12}
    \right\}.
\]
The equal-weight multilevel construction yields
\[
    \mathcal U_\alpha(\nu^\star)
    =
    \left\{
        \boldsymbol{\theta}\in\Theta:
        \theta_1+6\theta_2\leq\frac{5}{6}
    \right\}.
\]
Its area is
\[
    \mu\!\left(\mathcal U_\alpha(\nu^\star)\right)
    =
    \frac{25}{432}
    \approx 0.05787,
\]
whereas the smallest fixed-level set has area
\[
    \min\{\mu(\mathcal U_1),\mu(\mathcal U_2)\}
    =
    \frac{23}{288}
    \approx 0.07986.
\]
Thus, this setting provides an explicit example in which a nondegenerate
mixing distribution is strictly more efficient than every point-mass
construction.

For each mechanism, Figure~\ref{fig:tradeoff} averages set size over $60$
Monte Carlo repetitions, each using $3{,}000$ calibration observations. The
top panels consider $\alpha\in\{0.1,0.2,0.3\}$ and compare the point-mass
construction $\delta_\gamma$ with
\[
    \nu_{\gamma,\gamma'}
    =
    \frac{1}{2}\delta_\gamma+\frac{1}{2}\delta_{\gamma'},
\]
where the latter is minimized over $\gamma'$ for each displayed $\gamma$.
The bottom panels show the complete set-size surface over
$(\gamma,\gamma')$ at $\alpha=0.2$ and the corresponding minimizing path
$\gamma'^\star(\gamma)$.

\subsection{Parameter-Space Evaluation and Construction}
\label{app:forward-evaluation}
Fix a context $x$, calibration data $\mathcal D_n$, and a response-space miscoverage level $\gamma\in(0,\alpha)$. Recall that the forward compatibility is
\[
p_\gamma(\theta,x)=P_\theta\!\left\{Y\in\mathcal C_\gamma(x;\mathcal D_n)\mid x\right\}.
\]
Because the forward family $P_\theta(\cdot\mid x)$ is given, this probability can be evaluated directly by integrating or summing the forward law over $\mathcal C_\gamma(x;\mathcal D_n)$.

When $\Theta$ is finite, we enumerate every $\theta\in\Theta$ and apply the inclusion rule in \eqref{eq:lawcp-set}. When $\Theta$ is continuous, we evaluate the same rule over a finite parameter design $\Theta_N=\{\theta_1,\ldots,\theta_N\}$ and construct
\begin{equation}
\label{eq:grid-latent-set}
\widehat{\mathcal U}_{\alpha,N}(x;\gamma)
=
\left\{
\theta_j\in\Theta_N:
p_\gamma(\theta_j,x)\geq 1-\frac{\gamma}{\alpha}
\right\}.
\end{equation}
For a multilevel distribution $\nu=\sum_{k=1}^K w_k\delta_{\gamma_k}$, the same parameter design is used to compute
\[
e_\nu(\theta_j,x)
=
\sum_{k=1}^K
w_k\frac{1-p_{\gamma_k}(\theta_j,x)}{\gamma_k},
\]
yielding
\begin{equation}
\label{eq:grid-multilevel-set}
\widehat{\mathcal U}_{\alpha,N}(x;\nu)
=
\left\{
\theta_j\in\Theta_N:
e_\nu(\theta_j,x)\leq\frac{1}{\alpha}
\right\}.
\end{equation}
All forward-compatibility evaluations are independent across candidates and can therefore be parallelized.

We next describe an expansion of the finite-design set that preserves the coverage guarantee when $\Theta$ is continuous.
Let
\[
r_N
=
\sup_{\theta\in\Theta}
\min_{1\leq j\leq N}
\|\theta-\theta_j\|
\]
be the covering radius of $\Theta_N$. Suppose that, for each response level $\gamma_k$ used in the construction,
\[
|p_{\gamma_k}(\theta,x)-p_{\gamma_k}(\theta',x)|
\leq
L_{\gamma_k}(x)\|\theta-\theta'\|,
\qquad
\theta,\theta'\in\Theta .
\]
Then every $\theta\in\Theta$ lies in some ball $B(\theta_j,r_N)\cap\Theta$, and for all $\theta$ in this ball,
\[
p_{\gamma_k}(\theta,x)
\leq
u_{jk}(x)
\coloneqq
\min\{1,p_{\gamma_k}(\theta_j,x)+L_{\gamma_k}(x)r_N\}.
\]
For the multilevel construction, define the lower cell-wise incompatibility bound
\[
\underline e_{\nu,j}(x)
=
\sum_{k=1}^K
w_k\frac{1-u_{jk}(x)}{\gamma_k}.
\]
Since $\underline e_{\nu,j}(x)\leq e_\nu(\theta,x)$ for every $\theta\in B(\theta_j,r_N)\cap\Theta$, the conservative expansion is
\[
\widehat{\mathcal U}^{\mathrm{exp}}_{\alpha,N}(x;\nu)
=
\bigcup_{j:\,\underline e_{\nu,j}(x)\leq 1/\alpha}
\left\{
B(\theta_j,r_N)\cap\Theta
\right\}.
\]
For the point-mass case $\nu=\delta_\gamma$, this reduces to
\[
\widehat{\mathcal U}^{\mathrm{exp}}_{\alpha,N}(x;\gamma)
=
\bigcup_{j:\,u_j(x;\gamma)\geq 1-\gamma/\alpha}
\left\{
B(\theta_j,r_N)\cap\Theta
\right\},
\]
where $u_j(x;\gamma)=\min\{1,p_\gamma(\theta_j,x)+L_\gamma(x)r_N\}$.

\begin{proposition}[Coverage of the expanded construction]
\label{prop:expanded-construction-validity}
Suppose that $\Theta_N$ has covering radius $r_N$ and that the Lipschitz bounds above hold for all response levels used in the construction. Then, for any fixed multilevel distribution $\nu=\sum_{k=1}^K w_k\delta_{\gamma_k}$ satisfying the conditions of Theorem~\ref{thm:multilevel-validity},
\[
\mathcal U_\alpha(x;\nu)
\subseteq
\widehat{\mathcal U}^{\mathrm{exp}}_{\alpha,N}(x;\nu).
\]
Consequently,
\[
\mathbb P\left\{
\theta_{n+1}\in
\widehat{\mathcal U}^{\mathrm{exp}}_{\alpha,N}(X_{n+1};\nu)
\right\}
\geq
1-\alpha.
\]
In particular, taking $\nu=\delta_\gamma$ gives the same conclusion for the point-mass construction $\widehat{\mathcal U}^{\mathrm{exp}}_{\alpha,N}(x;\gamma)$.
\end{proposition}

\begin{proof}
Fix $\theta\in\mathcal U_\alpha(x;\nu)$. By the definition of the covering radius, there exists $j$ such that $\theta\in B(\theta_j,r_N)\cap\Theta$. For every $k$, the Lipschitz bound implies
\[
p_{\gamma_k}(\theta,x)
\leq
p_{\gamma_k}(\theta_j,x)+L_{\gamma_k}(x)r_N
\leq
u_{jk}(x).
\]
Therefore
\[
\underline e_{\nu,j}(x)
=
\sum_{k=1}^K
w_k\frac{1-u_{jk}(x)}{\gamma_k}
\leq
\sum_{k=1}^K
w_k\frac{1-p_{\gamma_k}(\theta,x)}{\gamma_k}
=
e_\nu(\theta,x).
\]
Since $\theta\in\mathcal U_\alpha(x;\nu)$, we have $e_\nu(\theta,x)\leq 1/\alpha$. Hence $\underline e_{\nu,j}(x)\leq 1/\alpha$, so the ball $B(\theta_j,r_N)\cap\Theta$ is included in $\widehat{\mathcal U}^{\mathrm{exp}}_{\alpha,N}(x;\nu)$. Thus $\theta\in\widehat{\mathcal U}^{\mathrm{exp}}_{\alpha,N}(x;\nu)$, proving the set inclusion.

The coverage statement follows immediately from Theorem~\ref{thm:multilevel-validity} and the inclusion
\[
\mathcal U_\alpha(X_{n+1};\nu)
\subseteq
\widehat{\mathcal U}^{\mathrm{exp}}_{\alpha,N}(X_{n+1};\nu).
\]
The point-mass case follows by setting $\nu=\delta_\gamma$, for which Theorem~\ref{thm:lawcp-validity} applies.
\end{proof}

\subsection{Optimization of Response Levels}
\label{app:gamma-tuning}
We now describe how to optimize the response-space levels for both the point-mass construction and the multilevel construction with a fixed number of levels $K$. The optimization is performed on an independent tuning split
\[
\mathcal D_m^{\mathrm{tun}}
=
\{(X_i^{\mathrm{tun}},Y_i^{\mathrm{tun}})\}_{i=1}^m,
\]
which is independent of the final calibration data $\mathcal D_n$. For a search interval $\Gamma=[\gamma_{\min},\gamma_{\max}]\subset(0,\alpha)$, define the breakpoint set
\begin{equation}
\label{eq:tuning-breakpoints}
\Gamma_m
=
\{\gamma_{\min}\}
\cup
\left\{
\frac{j}{m+1}:
\gamma_{\min}<\frac{j}{m+1}\leq\gamma_{\max},
\ j=1,\ldots,m
\right\}.
\end{equation}

For each $\gamma\in\Gamma_m$, construct provisional sets using $\mathcal D_m^{\mathrm{tun}}$ and define
\begin{equation}
\label{eq:empirical-size-objective}
\widehat R_m(\gamma)
=
\frac{1}{m}
\sum_{i=1}^m
\widehat\mu_N\!\left\{
\widehat{\mathcal U}_{\alpha,N}
(X_i^{\mathrm{tun}};\gamma,\mathcal D_m^{\mathrm{tun}})
\right\},
\end{equation}
where $\widehat\mu_N$ is the numerical approximation of the chosen set-size measure. For example, for a weighted parameter design,
\[
\widehat\mu_N(A)
=
\sum_{j=1}^N
a_j\mathbf 1\{\theta_j\in A\},
\]
with $a_j$ equal to the cell volume or another prespecified quadrature weight. We select
\begin{equation}
\label{eq:tuned-gamma}
\widehat\gamma
\in
\arg\min_{\gamma\in\Gamma_m}
\widehat R_m(\gamma),
\end{equation}
using a fixed rule, such as the smallest $\gamma$, to break ties.

The finite search over $\Gamma_m$ is exact for the empirical point-mass objective. Between consecutive breakpoints, the conformal rank $\lceil(m+1)(1-\gamma)\rceil$ is constant, so the response set is fixed. Over the same interval, the inclusion threshold $1-\gamma/\alpha$ decreases with $\gamma$, and hence the latent set can only expand. Therefore the minimum over each interval is attained at its left endpoint.

\begin{algorithm}[t]
\caption{Point-mass tuning of the response level}
\label{alg:point-mass-gamma-tuning}
\begin{algorithmic}[1]
\Require Independent tuning data $\mathcal D_m^{\mathrm{tun}}$; search interval $\Gamma$; parameter design $\Theta_N$ with size weights $\{a_j\}_{j=1}^N$; forward evaluator.
\State Construct $\Gamma_m$ from~\eqref{eq:tuning-breakpoints}.
\ForAll{$\gamma\in\Gamma_m$}
    \State Construct the tuning response set $\mathcal C_\gamma(\cdot;\mathcal D_m^{\mathrm{tun}})$.
    \For{$i=1,\ldots,m$}
        \State Evaluate $p_\gamma(\theta_j,X_i^{\mathrm{tun}})$ for all $\theta_j\in\Theta_N$.
        \State Construct $\widehat{\mathcal U}_{\alpha,N}(X_i^{\mathrm{tun}};\gamma,\mathcal D_m^{\mathrm{tun}})$.
    \EndFor
    \State Compute $\widehat R_m(\gamma)$ from~\eqref{eq:empirical-size-objective}.
\EndFor
\State \Return $\widehat\gamma$ from~\eqref{eq:tuned-gamma}.
\end{algorithmic}
\end{algorithm}

For the multilevel construction, let
\[
\nu_{\gamma,w}
=
\sum_{k=1}^K w_k\delta_{\gamma_k},
\qquad
w\in\Delta_K,
\]
where $\Delta_K=\{w\in[0,1]^K:\sum_{k=1}^K w_k=1\}$. We optimize over support locations $\gamma_1,\ldots,\gamma_K\in\Gamma_m$ and over a finite weight grid $\mathcal W_K\subset\Delta_K$. For a candidate pair $(\gamma,w)$, define
\[
\widehat R_m(\nu_{\gamma,w})
=
\frac{1}{m}
\sum_{i=1}^m
\widehat\mu_N\!\left\{
\widehat{\mathcal U}_{\alpha,N}
(X_i^{\mathrm{tun}};\nu_{\gamma,w},\mathcal D_m^{\mathrm{tun}})
\right\}.
\]
Equivalently, after computing
\[
e_{\gamma_k}(\theta_j,X_i^{\mathrm{tun}})
=
\frac{1-p_{\gamma_k}(\theta_j,X_i^{\mathrm{tun}})}{\gamma_k},
\]
the multilevel inclusion rule is
\[
\sum_{k=1}^K
w_k e_{\gamma_k}(\theta_j,X_i^{\mathrm{tun}})
\leq
\frac{1}{\alpha}.
\]
We select
\begin{equation}
\label{eq:tuned-multilevel}
(\widehat\gamma_{1:K},\widehat w)
\in
\arg\min_{\gamma_{1:K}\in\Gamma_m^K,\ w\in\mathcal W_K}
\widehat R_m(\nu_{\gamma,w}),
\end{equation}
with a fixed tie-breaking rule. One may restrict to ordered tuples $\gamma_1\leq\cdots\leq\gamma_K$ to remove duplicate representations.

\begin{algorithm}[t]
\caption{Multilevel tuning of response levels and weights}
\label{alg:multilevel-gamma-weight-tuning}
\begin{algorithmic}[1]
\Require Independent tuning data $\mathcal D_m^{\mathrm{tun}}$; breakpoint set $\Gamma_m$; number of levels $K$; weight grid $\mathcal W_K$; parameter design $\Theta_N$ with size weights $\{a_j\}_{j=1}^N$; forward evaluator.
\State Precompute $e_{\gamma}(\theta_j,X_i^{\mathrm{tun}})$ for all $\gamma\in\Gamma_m$, $j=1,\ldots,N$, and $i=1,\ldots,m$.
\ForAll{ordered tuples $\gamma_{1:K}\in\Gamma_m^K$}
    \ForAll{$w\in\mathcal W_K$}
        \For{$i=1,\ldots,m$}
            \State Construct
            \[
            \widehat{\mathcal U}_{\alpha,N}
            (X_i^{\mathrm{tun}};\nu_{\gamma,w},\mathcal D_m^{\mathrm{tun}})
            =
            \left\{
            \theta_j\in\Theta_N:
            \sum_{k=1}^K
            w_k e_{\gamma_k}(\theta_j,X_i^{\mathrm{tun}})
            \leq
            \frac{1}{\alpha}
            \right\}.
            \]
        \EndFor
        \State Compute $\widehat R_m(\nu_{\gamma,w})$.
    \EndFor
\EndFor
\State \Return $(\widehat\gamma_{1:K},\widehat w)$ from~\eqref{eq:tuned-multilevel}.
\end{algorithmic}
\end{algorithm}

Let $G=|\Gamma_m|$ and let $C_p$ denote the cost of one forward-compatibility evaluation. The point-mass tuning cost is $O(mNGC_p)$, plus $O(mNG)$ operations to form the sets and compute their sizes. For multilevel tuning, the precomputation cost is again $O(mNGC_p)$. A direct search over ordered $K$-tuples and a weight grid $\mathcal W_K$ then costs
\[
O\!\left(mN K {\binom{G+K-1}{K}}|\mathcal W_K|\right)
\]
after precomputation. In practice, $K$ is kept small, and the dominant computations over $i$, $j$, and $\gamma$ are parallelizable. Because tuning uses only $\mathcal D_m^{\mathrm{tun}}$, numerical approximation in the tuning objective affects efficiency but not the finite-sample validity of the final set calibrated on $\mathcal D_n$.

For a fixed response level, sorting the $n$ calibration scores costs
$O(n\log n)$ time and $O(n)$ memory; a linear-time order-statistic routine may be used
when only one level is required. If $N$ parameter candidates are evaluated using $M$
forward simulations each, the leading cost is $O(NM)$ time and $O(N)$ memory when
the simulations are processed sequentially. The $N$ compatibility evaluations are
embarrassingly parallel.

Let $K=|\Gamma_m|$ be the number of tuning breakpoints. A direct implementation over
$m$ tuning contexts costs $O(KmNM)$ Monte Carlo operations. This can be reduced by
reusing the same $M$ forward draws for every candidate level. After sorting the $M$
simulated nonconformity scores for each context--parameter pair, all $K$ compatibility
values are empirical distribution-function queries, giving total time
\[
    O\!\left(
        mN\{M\log M+K\log M\}
    \right),
\]
in addition to the $O(m\log m)$ cost of sorting the tuning calibration scores. The
calculation may be streamed over contexts or parameter candidates when storing all
$mNM$ simulated scores is impractical.

\section{Experiment Details}
\label{app:experiment}

This appendix describes the baseline methods implementations, data-generating mechanisms, sample splitting, and real-data details used in the experiments. Unless stated otherwise, every feasible method observes only context--response pairs $(X,Y)$ during fitting and calibration and receives only the context $X$ at test time.

\subsection{Baseline Methods}

Let $p_\theta(y\mid x)$ denote the density or probability mass function of the
forward law $P_\theta(\cdot\mid x)$, and let
$\Theta_N=\{\vartheta_1,\ldots,\vartheta_N\}$ be the common latent-parameter
grid when discretization is required. All feasible baselines are fitted using
only observed context--response pairs and receive only $X_{n+1}$ at test time.

\noindent\textbf{EmpBayes.}
We posit a parametric mixing law $Q_\eta(\mathrm d\theta\mid x)$ and estimate
$\eta$ from the marginal response model
\[
    p_\eta(y\mid x)
    =\int_\Theta p_\theta(y\mid x)Q_\eta(\mathrm d\theta\mid x).
\]
The continuous experiments use a Gaussian or diagonal-Gaussian mixing family,
fitted by marginal moments or EM as appropriate. The reported set is the
smallest grid-weighted region carrying at least $1-\alpha$ mass under
$Q_{\widehat\eta}(\cdot\mid x)$ \citep{efron1973stein,ignatiadis2019covariate}.

\noindent\textbf{NPMLE.}
This likelihood baseline estimates an unrestricted discrete mixing law on
$\Theta_N$ by
\[
    \widehat w\in\arg\max_{w\in\Delta_N}
    \sum_i\log\!\left\{
        \sum_{j=1}^N w_j p_{\vartheta_j}(Y_i\mid X_i)
    \right\},
\]
using EM. It returns a minimum-mass $1-\alpha$ grid set; for an ordered scalar
grid, we use the shortest interval carrying that mass. Thus, \texttt{NPMLE} denotes
the nonparametric maximum-likelihood estimator of the latent mixing law
\citep{lindsay1995mixture,soloff2025multivariate}.

\noindent\textbf{Noisy CP.}
For each observed pair, we first construct the unitwise inverse proxy
\[
    \widetilde\theta_i
    \in\arg\max_{\vartheta\in\Theta_N}
    p_\vartheta(Y_i\mid X_i),
\]
using the prescribed tie-breaking rule when the inverse is non-unique. A
predictor $\widehat m(x)$ is fitted to the training proxies, and split conformal
prediction is then calibrated using
$d\{\widetilde\theta_i,\widehat m(X_i)\}$. Its test set is
\[
    \mathcal U^{\mathrm{Noisy}}_\alpha(x)
    =\left\{\vartheta\in\Theta_N:
      d\{\vartheta,\widehat m(x)\}\leq
      \widehat q_{1-\alpha}\right\}.
\]
This baseline treats noisy inverse estimates as if they were latent labels and
therefore need not cover the true $\theta_{n+1}$ when inversion is ambiguous
\citep{lei2018distribution,einbinder2024label}.

\noindent\textbf{Oracle CP.}
Oracle CP applies the same split-conformal construction using the true latent
labels $\theta_i$ in place of $\widetilde\theta_i$. For the two-dimensional
grid experiments, the corresponding oracle reference is the smallest region
$A$ satisfying $\pi(A\mid x)\geq 1-\alpha$ under the true conditional latent
law. Both versions use information unavailable to feasible methods and serve
only as efficiency benchmarks.

\subsection{Synthetic Data Generation}
The synthetic experiments cover four scalar and four two-dimensional latent parameter families. In every setting, $X$ is observed, $\theta$ is latent, and $Y$ is generated from the known forward law $P_\theta(\cdot\mid X)$. The distribution of $\theta\mid X$ is used only to generate data and is not supplied to LatentCP. Each synthetic dataset is divided into independent training, tuning, calibration, and test samples. Unless stated otherwise, each split contains $1{,}000$ observations, and results are evaluated on an independent test sample. We use the absolute residual score for scalar responses, equivalently the $\ell_2$ score
$\widehat s(x,y)=\lVert y-\widehat f(x)\rVert_2$. The tuning sample selects $\gamma$ or the multilevel support and weights, after which the response sets are calibrated on the independent calibration sample.

\paragraph{Scalar $\theta$ data-generating mechanisms.}
The four one-dimensional settings in Figure~\ref{fig:res-coverage-eff} are:
\begin{itemize}[leftmargin=*,itemsep=2pt,topsep=3pt]
    \item \textbf{Gaussian mixture.} Draw $X\sim\mathcal N(0,I_3)$ and set
    $m(X)=X_1-0.6X_2+0.4X_3$. Then
    $\theta\mid X\sim0.85\mathcal N\{m(X)-0.3,0.2^2\}
    +0.15\mathcal N\{m(X)+1.7,0.2^2\}$ and
    $Y\mid\theta,X\sim\mathcal N(\theta,0.2^2)$.

    \item \textbf{Sign nonidentifiability.} Let
    \[
    \begin{aligned}
        Z\mid X&\sim\mathcal N\{1.2+0.25\tanh(X_1),0.1^2\},\\
        M&=\max\{0.1,Z\},\quad \theta=SM,\quad \Pr(S=1)=0.9.
    \end{aligned}
    \]
    The forward law is $Y\mid\theta,X\sim\mathcal N(\theta^2,0.25^2)$, so
    $\theta$ and $-\theta$ are observationally equivalent.

    \item \textbf{Regime crossing.} With no informative context, the latent
    law belongs to a core, displaced, or bursty family with total masses
    $(0.95,0.02,0.03)$. The core laws are
    $P_a=\mathcal N(a,1)$ for $a\in\{-1,-0.75,\ldots,1\}$, weighted by a
    discrete Gaussian kernel with scale $0.5$. The displaced laws are
    $P_{s,d}=\mathcal N(sd,0.1^2)$ for $s\in\{-1,1\}$ and
    $d\in\{2.5,2.75,3,3.25,3.5\}$. Finally, the bursty laws are
    $P_b=(1-b)\mathcal N(0,0.1^2)+(b/2)\mathcal N(-30,0.1^2)
    +(b/2)\mathcal N(30,0.1^2)$ for
    $b\in\{0.10,0.125,0.15,0.175,0.20\}$. The displaced and bursty laws are
    uniform within their respective families. Their response distributions
    cross repeatedly as the response-set level changes, making this setting
    favorable to combining information across multiple $\gamma$ values.

    \item \textbf{Aliased spikes.} Draw $X\sim\mathcal N(0,1)$ and assign
    $\theta$ to left, central, and right regimes with probabilities
    $(0.002,0.993,0.005)$. The outer regimes are uniform on $[-1,-0.1]$ and
    $[0.1,1]$; the central regime is split equally between width-$0.005$
    intervals adjacent to $-0.1$ and $0.1$. Let
    $(I_0,I_1,I_2,I_3)=([0,1],[2,3],[3,4],[4,5])$. Conditional on the left
    regime, $B$ is uniform on $\{0,3\}$; in the central regime, its
    probabilities over the four bands are $(938,15,35,5)/993$; and in the
    right regime, $B=1$. We then draw $S$ uniformly from $\{-1,1\}$,
    $U_B\sim\operatorname{Unif}(I_B)$, and set $Y=X+S U_B$.
\end{itemize}

\paragraph{Two-dimensional $\theta$ data-generating mechanisms.}
For all four families, $X\sim\mathcal N(0,I_2)$. We set
$\beta=(0.8,-0.4)^\top$ and $e(X)=\exp(0.35X_1)$. Conditional on $X=x$, the
latent coordinates are independent. Coordinate $j$ follows a mixture of two
components; its second component is selected with probability
\[
    w_j(x)=\operatorname{logit}^{-1}
    \{\operatorname{logit}(0.15)+0.9x_{d_j}\},
    \qquad (d_1,d_2)=(2,1).
\]
The component pairs and forward laws are:
\begin{itemize}[leftmargin=*,itemsep=2pt,topsep=3pt]
    \item \textbf{Gaussian location--scale.} With
    $\theta=(\mu,\log\sigma)$,
    \[
    \begin{aligned}
      \mu &: \mathcal N(0,0.30^2)\,/\,\mathcal N(0,0.70^2),\\
      \log\sigma &: \mathcal N(0,0.18^2)\,/\,
                     \mathcal N\{\log(1.6),0.18^2\},\\
      Y&=\beta^\top X+\mu+\sigma e(X)Z,\quad Z\sim\mathcal N(0,1).
    \end{aligned}
    \]

    \item \textbf{Negative binomial.} With
    $\theta=(\log\lambda,\log\phi)$,
    \[
    \begin{aligned}
      \log\lambda &: \mathcal N(0,0.10^2)\,/\,
                     \mathcal N\{\log(1.5),0.12^2\},\\
      \log\phi &: \mathcal N\{\log(35),0.20^2\}\,/\,
                  \mathcal N\{\log(5),0.15^2\},\\
      Y&\sim\operatorname{NegBin}(m,\phi),\quad
      m=6\exp\{\log\lambda+\beta^\top X\},\quad
      \operatorname{Var}(Y\mid X,\theta)=m+m^2/\phi.
    \end{aligned}
    \]

    \item \textbf{Beta-binomial.} With $\theta=(p,\log\kappa)$, the latent baseline probability $p$ is generated from a two-component mixture on the logit scale:
    \[
    \operatorname{logit}(p)
    \sim
    \pi\,\mathcal N\{\operatorname{logit}(0.50),\tau_{0.50}^2\}
    +
    (1-\pi)\,\mathcal N\{\operatorname{logit}(0.65),\tau_{0.65}^2\},
    \]
    where $\tau_{0.50}$ and $\tau_{0.65}$ are chosen so that the induced standard deviations of $p$ are approximately $0.03$ in both components. The remaining latent component and forward law are
    \[
    \begin{aligned}
      \log\kappa &: \mathcal N\{\log(220),0.20^2\}\,/\,
                    \mathcal N\{\log(40),0.15^2\},\\
      \widetilde p&=\operatorname{logit}^{-1}
        \{\operatorname{logit}(p)+0.5\beta^\top X\},\\
      Y&\sim\operatorname{BetaBin}(40,\widetilde p,\kappa).
    \end{aligned}
    \]

    \item \textbf{Student $t$.} With $\theta=(\log\sigma,\log\nu)$, the
    component pairs and forward law are
    \[
    \begin{aligned}
      \log\sigma &: \mathcal N(0,0.15^2)\,/\,
                     \mathcal N\{\log(1.5),0.15^2\},\\
      \log\nu &: \mathcal N\{\log(35),0.20^2\}\,/\,
                  \mathcal N\{\log(6),0.12^2\},\\
      Y&=\beta^\top X+\sigma e(X)T_\nu,\quad T_\nu\sim t_\nu.
    \end{aligned}
    \]
\end{itemize}
The geometry panels use $x=(1.5,1.5)$ and a $40\times40$ parameter grid; all
mixture draws are truncated to the candidate rectangles used in the code.

\subsection{Real Data}

We use $1,689$ California wildfires recorded by Monitoring Trends in Burn Severity during 1984--2023 and 2024 Census county boundaries \citep{eidenshink2007project,finco2012monitoring}. We partition the study
period into ten nonoverlapping four-year windows. Figure~\ref{fig:setting}
uses California's 58 counties (580 county--window observations) for an interpretable geographic illustration. The formal analysis in Figure~\ref{fig:wildfire-overlay} uses 179 equal-area 50-km grid cells clipped to the state boundary (1,790 cell--window observations), retaining cells with at least 15\% of their original area.

For spatial unit $s$ and window $t$, let $Y_{s,t}$ denote the mapped-fire count and define the area--time exposure
\[
    \Delta_{s,t}
    =
    \frac{\operatorname{area}(s)}{1{,}000\ \mathrm{km}^2}
    \times 4\ \mathrm{years}.
\]
We use the forward model
\[
    Y_{s,t}\mid X_{s,t},\lambda_{s,t}
    \sim
    \operatorname{Poisson}(\Delta_{s,t}\lambda_{s,t}),
\]
so $\lambda_{s,t}$ is the latent annual intensity per
$1{,}000\ \mathrm{km}^2$. The context $X_{s,t}\in\mathbb R^{10}$ contains standardized east--west and north--south coordinates, standardized time, log exposure, retained area fraction, transformed preceding-window count, historical rate computed strictly from earlier windows, and three spatial--temporal interaction terms.

We fit a gradient-boosting model for the observed count on $1984$--$2003$, tune $\gamma$ on $2004$--$2007$, calibrate on $2008$--$2015$, and evaluate on $2016$--$2023$. With $\alpha=0.10$, we select $\gamma$ by minimizing mean latent-set size on the tuning window and then calibrate the response set on the calibration data. For each test context and exposure, the reported set is
\[
    \mathcal U_\alpha(X)
    =
    \left\{
        \lambda:
        \mathbb P_{Y\sim\operatorname{Poisson}(\Delta\lambda)}
        \left\{Y\notin\mathcal C_{\widehat\gamma}(X)\right\}
        \leq \frac{\widehat\gamma}{\alpha}
    \right\}.
\]

For the grid analysis, tuning selects $\widehat\gamma=0.0278$. Across the $358$ held-out cell--window observations, the response sets attain empirical coverage $0.992$, and tuning reduces mean latent-set size from $0.551$ under the fixed choice $\gamma=\alpha/2$ to $0.490$, an $11.1$\% reduction. The optimized two-level construction selects the same level at both support points and therefore coincides with the tuned single-level method in this application. The latent maps use a common nonlinear color scale only to reveal small positive intensities; it does not alter the reported endpoints or widths. Because $\lambda$ is unobserved and the observations exhibit spatial and temporal dependence, we report response coverage and latent-set size, but not empirical latent coverage.
\end{document}